\documentclass[journal,a4paper,12pt,onecolumn]{IEEEtran}
\usepackage{amssymb}
\usepackage{amsmath}
\usepackage{amsfonts}
\usepackage{graphicx}
\usepackage{subfigure}

\newtheorem{theorem}{Theorem}

\newtheorem{algorithm}[theorem]{Algorithm}

\newtheorem{definition}[theorem]{Definition}
\newtheorem{example}[theorem]{Example}

\newtheorem{lemma}[theorem]{Lemma}

\newtheorem{remark}[theorem]{Remark}

\begin{document}

\title{On Near-controllability, Nearly-controllable Subspaces, and
Near-controllability Index of a Class of Discrete-time Bilinear Systems: A
Root Locus Approach\thanks{%
This work was supported by the China Postdoctoral Science Foundation funded
project and the National Natural Science Foundation of China.}}
\author{Lin Tie\thanks{%
The author is with the School of Automation Science and Electrical
Engineering, Beihang University (Beijing University of Aeronautics and
Astronautics), 100191, Beijing, P. R. China. \textit{E-mail address}:
tielinllc@gmail.com.}}
\maketitle

\begin{abstract}
This paper studies near-controllability of a class of discrete-time bilinear
systems via a root locus approach. A necessary and sufficient criterion for
the systems to be nearly controllable is given. In particular, by using the
root locus approach, the control inputs which achieve the state transition
for the nearly controllable systems can be computed. Furthermore, for the
non-nearly controllable systems, nearly-controllable subspaces are derived
and near-controllability index is defined. Accordingly, the controllability
properties of such class of discrete-time bilinear systems are fully
characterized. Finally, examples are provided to demonstrate the results of
the paper.
\end{abstract}


\section{Introduction}

Given the capability of modeling a large number of processes in the real
world, bilinear systems have received considerable attention [1-4].
Furthermore, bilinear systems are thought to be simpler and better
understood than most other nonlinear systems. Such systems have particular
advantages on structural properties, optimization, identification, and
control. The relaxed version of a linear switched system is actually a
bilinear system [5]. Owing to both the practical and theoretical importance,
bilinear systems have been a hot topic in the literature of nonlinear
systems over decades.

$\left. {}\right. $

Controllability is one fundamental concept in mathematical control theory.
It was identified in the early 1960s, and then the theory of controllability
for linear systems based on the state space description was systematically
established [6,7]. During nearly the same period, controllability of
nonlinear systems also was considered. Since the 1970s, Lie algebra methods
and other powerful tools of differentiable manifold theory have been
developed to study controllability of nonlinear systems [8-11]. Today,
controllability has played an essential role in the development and
application of mathematical control theory. There are many different kinds
of definitions on controllability, such as local controllability, global
controllability, approximate controllability, positive controllability, and
null controllability. Roughly speaking, controllability is defined as the
ability of a system that the system can be steered from an arbitrary initial
state to an arbitrary terminal state under the action of admissible
controls. It is an important as well as fundamental\ property of control
systems and is of great practical relevance.

$\left. {}\right. $

In bilinear system theory, controllability is one of the main research
topics. This is particularly true for the continuous-time case. More
specifically, controllability of continuous-time bilinear systems has been
extensively investigated profiting from the Lie algebra methods. Various
Lie-algebraic criteria on controllability of continuous-time bilinear
systems were obtained in the literature, which have been summarized and
updated in a recent monograph [4] on bilinear systems. However, for
discrete-time bilinear systems, the controllability results are rather
sparse compared with their continuous-time counterparts. Most of the works
on controllability of discrete-time bilinear systems were done in the 1970s
[12-15], dealing with systems of the form%
\begin{equation}
x\left( k+1\right) =\left( A+u\left( k\right) B\right) x\left( k\right)
\end{equation}%
where $x(k)\in
\mathbb{R}
^{n}$, $A,B\in
\mathbb{R}
^{n\times n}$, and $u(k)\in
\mathbb{R}
$. \textit{System (1) is said to be controllable if, for any }$\xi ,\eta $%
\textit{\ in }$%
\mathbb{R}
_{\ast }^{n}$\textit{\ (}$%
\mathbb{R}
_{\ast }^{n}:=%
\mathbb{R}
^{n}\setminus \left\{ 0\right\} $\textit{),\ there exist a positive integer }%
$l$\textit{\ and a finite control sequence }$u(k)$\textit{\ (}$k=0,1,\ldots
,l-1$\textit{) such that the system can be steered from }$\xi $\textit{\ to }%
$\eta $\textit{\ at }$k=l$. In particular, [12] gave a sufficient condition
for controllability of system (1), which requires, at least, $A$\ is similar
to an orthogonal matrix. [13] studied controllability of system (1) under
the assumption of rank$B=1$ and presented necessary as well as sufficient
conditions; based on the work in [13], [15] improved these conditions by
raising necessary and sufficient ones. It can be seen that, for
controllability of discrete-time bilinear systems, only specific subclasses
have been considered, while most cases remain unsolved. Since the middle
1980s, few work has been reported until the 2000s. One of the reasons
leading to such a few results on controllability of discrete-time bilinear
systems is that the controllability problems are quite difficult to deal
with due to the systems' nonlinearity. Another one is that, for
discrete-time nonlinear systems, semigroups tend to appear so that less
algebraic structure of the systems is available [16]. Recently, [17]
continued the study on controllability of system (1) and obtained necessary
and sufficient conditions for the case of $AB=BA$, where it was shown that
the system can be controllable only if its dimension is no greater than two.
Nevertheless, an uncontrollable system can still have a large \textit{%
controllable region}\footnote{%
A controllable region is a region in $%
\mathbb{R}
^{n}$ on which the system is controllable. Namely, for any $\xi ,\eta $ in
this region, there exist control inputs that steer the system from $\xi $ to
$\eta $.}. This has first been proved for system (1) with $A=I$, i.e.%
\begin{equation}
x\left( k+1\right) =\left( I+u\left( k\right) B\right) x\left( k\right)
\end{equation}%
which is uncontrollable with dimension greater than two [17]. Indeed, if $B$
has only real eigenvalues that are nonzero and pairwise distinct, then the
system (2) has a large controllable region, which nearly covers the whole
space, and is nearly controllable.

$\left. {}\right. $

If we only use \textquotedblleft uncontrollable\textquotedblright\ to
describe a system which is not controllable according to the general
controllability definition, we may miss some valuable properties of it.
Near-controllability is thus introduced to describe those systems that are
uncontrollable but have a very large controllable region. Nearly
controllable systems exist widely in nonlinear systems. This property was
first defined and was demonstrated on system (2) in [18], and it was then
generalized to system (1) with $B=I$ in [19] and to both continuous-time and
discrete-time nonlinear systems that are not necessarily bilinear in [20].
The definition of near-controllability, which was first given in [18], has
now been updated in [20]. \textit{A continuous-time system }$\dot{x}\left(
t\right) =f\left( x\left( t\right) ,u\left( t\right) \right) $\textit{\
(discrete-time system }$x\left( k+1\right) =f\left( x\left( k\right)
,u\left( k\right) \right) $\textit{) is said to be nearly controllable if,
for any }$\xi \in
\mathbb{R}
^{n}\left\backslash \mathcal{E}\right. $\textit{\ and any }$\eta \in
\mathbb{R}
^{n}\left\backslash \mathcal{F}\right. $\textit{, there exist a piecewise
continuous control }$u\left( t\right) $\textit{\ and }$T>0$\textit{\ (a
finite control sequence }$u\left( k\right) $\textit{, }$k=0,1,\ldots ,l-1$%
\textit{, where }$l$\textit{\ is a positive integer) such that the system
can be steered from }$\xi $\textit{\ to }$\eta $\textit{\ at some }$t\in
\left( 0,T\right) $\textit{\ (}$k=l$\textit{), where }$\mathcal{E}$\textit{\
and }$\mathcal{F}$\textit{\ are two sets of zero Lebesgue measure in }$%
\mathbb{R}
^{n}$. If we let $\mathcal{E},\mathcal{F}=\varnothing $, then the
near-controllability definition degenerates to the general controllability
definition. Thus, near-controllability includes the notion of
controllability and may better characterize the controllability properties
of nonlinear systems. However, most of the existing works on
near-controllability are reported for discrete-time bilinear systems and the
study on this topic is just at the beginning.

$\left. {}\right. $

In this paper, we continue the research on near-controllability of system
(2) with $B$ having real eigenvalues only. We improve the sufficient
condition for near-controllability of the system (2) in [18] by giving a
necessary and sufficient one. In particular, we apply a new approach this
time to prove near-controllability. That is, a root locus approach is
proposed in this paper. Compared with the technique used in [18,19] which is
based on the implicit function theorem, by the root locus approach we can
not only improve the obtained result on near-controllability, but also
compute the required control inputs that achieve the state transition%
\footnote{%
For nonlinear systems, it is, in general, hard to compute the control inputs
to achieve state transition even if controllability has been proved. In this
paper, we obtain both the controllability and \textquotedblleft control
computability\textquotedblright\ (the ability of computing the required
control inputs for the transition of any given initial and terminal states).}%
. We thus present a useful algorithm. Furthermore, inspired by the state
space description for linear systems, we also consider the non-nearly
controllable systems. It is well known that if a linear time-invariant
system is uncontrollable, then the state space can be decomposed as a direct
sum of a controllable subspace and an uncontrollable subspace. For the
non-nearly controllable systems, we derive nearly-controllable subspaces and
define near-controllability index by using the improved near-controllability
result, which shows that the non-nearly controllable systems can be
controllable on the nearly-controllable subspaces. The near-controllability
index is used to determine the largest nearly-controllable subspaces that a
non-nearly controllable system has. In summary, the controllability
properties of the system (2) are fully characterized. Finally, we provide
examples to illustrate the conceptions and the results of this paper.

$\left. {}\right. $

\section{A Root Locus Approach to Near-controllability}

In this section, we propose a root locus approach and apply it to prove
near-controllability of the system (2). The idea of the root locus approach
is to use the root locus theory to achieve the state transition, including
transferring any initial state to itself and to a state close to it. More
importantly, the control inputs are computable in these steps. Then, we use
the matrix theory by establishing a transition matrix and some algebraic
techniques to prove near-controllability. The proof steps are similar to
those in [18,19], but the factor that plays the key role has changed. That
is, the implicit function theorem has been replaced by the root locus
approach.

$\left. {}\right. $

\begin{lemma}
If $B\in
\mathbb{R}
^{n\times n}$ has only nonzero and real eigenvalues, then there exist
nonzero real numbers $u\left( 0\right) ,u\left( 1\right) ,\ldots ,u\left(
L-1\right) ,u\left( L\right) $ such that%
\begin{equation}
\left( I+u\left( L\right) B\right) \left( I+u\left( L-1\right) B\right)
\cdots \left( I+u\left( 1\right) B\right) \left( I+u\left( 0\right) B\right)
=I
\end{equation}%
if and only if the dimension of the largest Jordan block in the Jordan
canonical form of $B$ is no greater than two. Further, if $B$ not only
satisfies the necessary and sufficient condition but also is cyclic\footnote{%
A matrix is said to be cyclic if its characteristic polynomial is equal to
its minimal polynomial. Namely, only one Jordan block exists for each
eigenvalue in the Jordan canonical form of the matrix.}, then there exist $%
\left( 2m+3\right) $ nonzero real and pairwise distinct numbers $u\left(
0\right) ,u\left( 1\right) ,\ldots ,u\left( 2m+1\right) ,u\left( 2m+2\right)
$ such that (3) holds, where $m$ is the number of the distinct eigenvalues
of $B$.
\end{lemma}

\begin{proof}
We put the proof of necessity in Appendix and only prove sufficiency here.
Assume first that $B$ has the following cyclic form%
\begin{equation}
\left[
\begin{array}{cccccccc}
\lambda _{1} & 1 &  &  &  &  &  &  \\
& \lambda _{1} &  &  &  &  &  &  \\
&  & \ddots &  &  &  &  &  \\
&  &  & \lambda _{r} & 1 &  &  &  \\
&  &  &  & \lambda _{r} &  &  &  \\
&  &  &  &  & \lambda _{r+1} &  &  \\
&  &  &  &  &  & \ddots &  \\
&  &  &  &  &  &  & \lambda _{m}%
\end{array}%
\right]
\end{equation}%
where $B$ is already in Jordan canonical form since a nonsingular
transformation to $B$ does not affect the proof and $\lambda _{1},\ldots
,\lambda _{m}$ are nonzero real and pairwise distinct. In addition, $m+r=n$.
We now show that the equation%
\begin{equation}
\left( I+u\left( 2m+2\right) B\right) \left( I+u\left( 2m+1\right) B\right)
\cdots \left( I+u\left( 1\right) B\right) \left( I+u\left( 0\right) B\right)
=I
\end{equation}%
admits a solution of nonzero real $u\left( 0\right) ,u\left( 1\right)
,\ldots ,u\left( 2m+1\right) ,u\left( 2m+2\right) $. To this end, it will be
proved that the reciprocals of $u\left( 0\right) ,u\left( 1\right) ,\ldots
,u\left( 2m+1\right) ,u\left( 2m+2\right) $ are on the root loci of the
characteristic equation of a closed transfer function related to $B$'s
eigenvalues. Multiplying both sides of equation (5) by%
\begin{equation*}
\overset{2m+2}{\underset{k=0}{\prod }}\frac{1}{u\left( k\right) }
\end{equation*}%
we have%
\begin{equation}
\left( \frac{1}{u\left( 2m+2\right) }I+B\right) \left( \frac{1}{u\left(
2m+1\right) }I+B\right) \cdots \left( \frac{1}{u\left( 1\right) }I+B\right)
\left( \frac{1}{u\left( 0\right) }I+B\right) =\overset{2m+2}{\underset{k=0}{%
\prod }}\frac{1}{u\left( k\right) }I,
\end{equation}%
from which\footnote{%
Throughout this paper, such kind of expression $\overset{2m+2}{\underset{k=0}%
{\sum }}\frac{u\left( k\right) }{\overset{2m+2}{\underset{j=0}{\prod }}%
u\left( j\right) }$ is used to represent the summation
\par
\begin{equation*}
\sum \frac{1}{u\left( k_{0}\right) }\frac{1}{u\left( k_{1}\right) }\cdots
\frac{1}{u\left( k_{2m+1}\right) }=\frac{1}{u\left( 0\right) }\frac{1}{%
u\left( 1\right) }\cdots \frac{1}{u\left( 2m+1\right) }+\frac{1}{u\left(
0\right) }\cdots \frac{1}{u\left( 2m\right) }\frac{1}{u\left( 2m+2\right) }%
+\cdots +\frac{1}{u\left( 1\right) }\cdots \frac{1}{u\left( 2m+1\right) }%
\frac{1}{u\left( 2m+2\right) }
\end{equation*}%
for simplicity. There is no meaning of division in the expression.}%
\begin{equation}
B^{2m+3}+\overset{2m+2}{\underset{k=0}{\sum }}\frac{1}{u\left( k\right) }%
B^{2m+2}+\cdots +\overset{2m+2}{\underset{k=0}{\sum }}\frac{u\left( k\right)
}{\overset{2m+2}{\underset{j=0}{\prod }}u\left( j\right) }B=0.
\end{equation}%
Writing (7) through two groups of equations yields%
\begin{eqnarray}
&&\left\{
\begin{array}{c}
\lambda _{1}^{2m+3}+\overset{2m+2}{\underset{k=0}{\sum }}\frac{1}{u\left(
k\right) }\lambda _{1}^{2m+2}+\cdots +\overset{2m+2}{\underset{k=0}{\sum }}%
\frac{u\left( k\right) }{\overset{2m+2}{\underset{j=0}{\prod }}u\left(
j\right) }\lambda _{1}=0 \\
\left( 2m+3\right) \lambda _{1}^{2m+2}+\left( 2m+2\right) \overset{2m+2}{%
\underset{k=0}{\sum }}\frac{1}{u\left( k\right) }\lambda _{1}^{2m+1}+\cdots +%
\overset{2m+2}{\underset{k=0}{\sum }}\frac{u\left( k\right) }{\overset{2m+2}{%
\underset{j=0}{\prod }}u\left( j\right) }=0 \\
\vdots \\
\lambda _{r}^{2m+3}+\overset{2m+2}{\underset{k=0}{\sum }}\frac{1}{u\left(
k\right) }\lambda _{r}^{2m+2}+\cdots +\overset{2m+2}{\underset{k=0}{\sum }}%
\frac{u\left( k\right) }{\overset{2m+2}{\underset{j=0}{\prod }}u\left(
j\right) }\lambda _{r}=0 \\
\left( 2m+3\right) \lambda _{r}^{2m+2}+\left( 2m+2\right) \overset{2m+2}{%
\underset{k=0}{\sum }}\frac{1}{u\left( k\right) }\lambda _{r}^{2m+1}+\cdots +%
\overset{2m+2}{\underset{k=0}{\sum }}\frac{u\left( k\right) }{\overset{2m+2}{%
\underset{j=0}{\prod }}u\left( j\right) }=0%
\end{array}%
\right. ,  \notag \\
&&\left\{
\begin{array}{c}
\lambda _{r+1}^{2m+3}+\overset{2m+2}{\underset{k=0}{\sum }}\frac{1}{u\left(
k\right) }\lambda _{r+1}^{2m+2}+\cdots +\overset{2m+2}{\underset{k=0}{\sum }}%
\frac{u\left( k\right) }{\overset{2m+2}{\underset{j=0}{\prod }}u\left(
j\right) }\lambda _{r+1}=0 \\
\vdots \\
\lambda _{m}^{2m+3}+\overset{2m+2}{\underset{k=0}{\sum }}\frac{1}{u\left(
k\right) }\lambda _{m}^{2m+2}+\cdots +\overset{2m+2}{\underset{k=0}{\sum }}%
\frac{u\left( k\right) }{\overset{2m+2}{\underset{j=0}{\prod }}u\left(
j\right) }\lambda _{m}=0%
\end{array}%
\right. .
\end{eqnarray}%
By adding the following $\left( m-r\right) $\ constraints%
\begin{gather*}
\left( 2m+3\right) \lambda _{r+1}^{2m+2}+\left( 2m+2\right) \overset{2m+2}{%
\underset{k=0}{\sum }}\frac{1}{u\left( k\right) }\lambda
_{r+1}^{2m+1}+\cdots +\overset{2m+2}{\underset{k=0}{\sum }}\frac{u\left(
k\right) }{\overset{2m+2}{\underset{j=0}{\prod }}u\left( j\right) }=0, \\
\underset{}{\overset{}{\vdots }} \\
\left( 2m+3\right) \lambda _{m}^{2m+2}+\left( 2m+2\right) \overset{2m+2}{%
\underset{k=0}{\sum }}\frac{1}{u\left( k\right) }\lambda _{m}^{2m+1}+\cdots +%
\overset{2m+2}{\underset{k=0}{\sum }}\frac{u\left( k\right) }{\overset{2m+2}{%
\underset{j=0}{\prod }}u\left( j\right) }=0
\end{gather*}%
to the second group of equations in (8) and the two constraints%
\begin{eqnarray*}
\lambda _{m+1}^{2m+3}+\overset{2m+2}{\underset{k=0}{\sum }}\frac{1}{u\left(
k\right) }\lambda _{m+1}^{2m+2}+\cdots +\overset{2m+2}{\underset{k=0}{\sum }}%
\frac{u\left( k\right) }{\overset{2m+2}{\underset{j=0}{\prod }}u\left(
j\right) }\lambda _{m+1} &=&0, \\
\lambda _{m+2}^{2m+3}+\overset{2m+2}{\underset{k=0}{\sum }}\frac{1}{u\left(
k\right) }\lambda _{m+2}^{2m+2}+\cdots +\overset{2m+2}{\underset{k=0}{\sum }}%
\frac{u\left( k\right) }{\overset{2m+2}{\underset{j=0}{\prod }}u\left(
j\right) }\lambda _{m+2} &=&0
\end{eqnarray*}%
to all equations in (8), where $\lambda _{m+1},\lambda _{m+2}$ are chosen
such that%
\begin{equation}
0<\left\vert \lambda _{m+1}\right\vert <\min \left\{ \left\vert \lambda
_{1}\right\vert ,\ldots ,\left\vert \lambda _{m}\right\vert \right\} <\max
\left\{ \left\vert \lambda _{1}\right\vert ,\ldots ,\left\vert \lambda
_{m}\right\vert \right\} <-\lambda _{m+2},
\end{equation}%
we can put the above $\left( m+r\right) +\left( m-r\right) +2=2m+2$
equations into the matrix form%
\begin{equation*}
\left[
\begin{array}{ccc}
\lambda _{1}^{2m+2} & \cdots & \lambda _{1} \\
\left( 2m+2\right) \lambda _{1}^{2m+1} & \cdots & 1 \\
\vdots & \vdots & \vdots \\
\lambda _{m}^{2m+2} & \cdots & \lambda _{m} \\
\left( 2m+2\right) \lambda _{m}^{2m+1} & \cdots & 1 \\
\lambda _{m+1}^{2m+2} & \cdots & \lambda _{m+1} \\
\lambda _{m+2}^{2m+2} & \cdots & \lambda _{m+2}%
\end{array}%
\right] \left[
\begin{array}{c}
\overset{2m+2}{\underset{k=0}{\sum }}\frac{1}{u\left( k\right) } \\
\underset{}{\overset{}{\vdots }} \\
\overset{2m+2}{\underset{k=0}{\sum }}\frac{u\left( k\right) }{\overset{2m+2}{%
\underset{j=0}{\prod }}u\left( j\right) }%
\end{array}%
\right] =\left[
\begin{array}{c}
-\lambda _{1}^{2m+3} \\
-\left( 2m+3\right) \lambda _{1}^{2m+2} \\
\vdots \\
-\lambda _{m}^{2m+3} \\
-\left( 2m+3\right) \lambda _{m}^{2m+2} \\
-\lambda _{m+1}^{2m+3} \\
-\lambda _{m+2}^{2m+3}%
\end{array}%
\right] .
\end{equation*}%
From Lemma 15 in Appendix, we obtain%
\begin{eqnarray*}
\overset{2m+2}{\underset{k=0}{\sum }}\frac{1}{u\left( k\right) } &=&\left(
-1\right) \left( 2\overset{m}{\underset{i=1}{\sum }}\lambda _{i}+\lambda
_{m+1}+\lambda _{m+2}\right) , \\
&&\vdots \\
\overset{2m+2}{\underset{k=0}{\sum }}\frac{u\left( k\right) }{\overset{2m+2}{%
\underset{j=0}{\prod }}u\left( j\right) } &=&\left( -1\right) ^{2m+2}\lambda
_{m+1}\lambda _{m+2}\overset{m}{\underset{i=1}{\prod }}\lambda _{i}^{2}.
\end{eqnarray*}%
If we denote $\overset{2m+2}{\underset{k=0}{\prod }}\frac{1}{u\left(
k\right) }$ by%
\begin{equation*}
\left( -1\right) ^{2m+3}K,
\end{equation*}%
then, from the Vi\`{e}te's formulas, $\frac{1}{u\left( 0\right) },\frac{1}{%
u\left( 1\right) },\ldots ,\frac{1}{u\left( 2m+1\right) },\frac{1}{u\left(
2m+2\right) }$\ are the roots of the following$\ \left( 2m+3\right) $%
th-degree equation%
\begin{equation*}
s^{2m+3}+\left( 2\overset{m}{\underset{i=1}{\sum }}\lambda _{i}+\lambda
_{m+1}+\lambda _{m+2}\right) s^{2m+2}+\cdots +\left( \lambda _{m+1}\lambda
_{m+2}\overset{m}{\underset{i=1}{\prod }}\lambda _{i}^{2}\right) s+K=0
\end{equation*}%
and are thus on the root loci of the characteristic equation of the closed
loop transfer function%
\begin{eqnarray*}
1+KG\left( s\right) &:&=1+\frac{K}{s^{2m+3}+\left( 2\overset{m}{\underset{i=1%
}{\sum }}\lambda _{i}+\lambda _{m+1}+\lambda _{m+2}\right) s^{2m+2}+\cdots
+\left( \lambda _{m+1}\lambda _{m+2}\overset{m}{\underset{i=1}{\prod }}%
\lambda _{i}^{2}\right) s} \\
&=&1+\frac{K}{s(s+\lambda _{1})^{2}\cdots (s+\lambda _{m})^{2}(s+\lambda
_{m+1})(s+\lambda _{m+2})}.
\end{eqnarray*}%
$G\left( s\right) $ has $\left( 2m+3\right) $\ real poles only and has no
zero. By condition (9), in the $\left( 2m+3\right) $ real poles, $-\lambda
_{m+2}$ is the largest pole of $G\left( s\right) $ and is a single pole; $%
-\lambda _{1},\ldots ,-\lambda _{m}$ are double poles of $G\left( s\right) $%
; $0$ and $-\lambda _{m+1}$ are two single poles of $G\left( s\right) $
between two double poles in $-\lambda _{1},\ldots ,-\lambda _{m}$.
Therefore, we can see from the root locus theory [21] that, as $K$ increases
from $0$ to $+\infty $, all root loci of $1+KG\left( s\right) =0$ that start
at the $\left( 2m+3\right) $ real poles will first move on the real axis.
That is to say, we can always choose a $K$ such that $1+KG\left( s\right) =0$
has $\left( 2m+3\right) $ \textit{nonzero real} and \textit{pairwise distinct%
} roots. Indeed, by making $K$ small enough, $\frac{1}{u\left( 0\right) },%
\frac{1}{u\left( 1\right) },\ldots ,\frac{1}{u\left( 2m+1\right) },\frac{1}{%
u\left( 2m+2\right) }$ are perturbed away from $-\lambda _{1},-\lambda
_{1},\ldots ,-\lambda _{m},-\lambda _{m},-\lambda _{m+1},-\lambda _{m+2},0$
and hence are not only nonzero real but also pairwise distinct. Then, the
reciprocals of the nonzero real and pairwise distinct roots are the real
numbers that satisfy equation (5) since $\frac{1}{u\left( 0\right) },\frac{1%
}{u\left( 1\right) },\ldots ,\frac{1}{u\left( 2m+1\right) },\frac{1}{u\left(
2m+2\right) }$ satisfy the equations in (8), (7), and (6).

Finally, if $B$ is noncyclic, we can choose the largest cyclic main
submatrix of $B$, named $B_{sub}$, and there exist real numbers $u\left(
0\right) ,u\left( 1\right) ,\ldots ,u\left( L-1\right) ,u\left( L\right) $
such that%
\begin{equation}
\left( I+u\left( L\right) B_{sub}\right) \left( I+u\left( L-1\right)
B_{sub}\right) \cdots \left( I+u\left( 1\right) B_{sub}\right) \left(
I+u\left( 0\right) B_{sub}\right) =I
\end{equation}%
from the above analysis. By noting the fact that%
\begin{equation*}
\left( I+u\left( L\right) \left[
\begin{array}{cc}
\lambda & 1 \\
0 & \lambda%
\end{array}%
\right] \right) \cdots \left( I+u\left( 0\right) \left[
\begin{array}{cc}
\lambda & 1 \\
0 & \lambda%
\end{array}%
\right] \right) =I
\end{equation*}%
implies%
\begin{equation*}
\left( I+u\left( L\right) \left[
\begin{array}{ccc}
\lambda & 1 & 0 \\
0 & \lambda & 0 \\
0 & 0 & \lambda%
\end{array}%
\right] \right) \cdots \left( I+u\left( 0\right) \left[
\begin{array}{ccc}
\lambda & 1 & 0 \\
0 & \lambda & 0 \\
0 & 0 & \lambda%
\end{array}%
\right] \right) =I,
\end{equation*}%
one can easily verify that eq. (10) still works if $B_{sub}$ is replaced by $%
B$.
\end{proof}

$\left. {}\right. $

By using Lemma 1, we can transfer an arbitrary state to itself and improve
Theorem 2 obtained in [18]. Moreover, by the root locus approach, we can
compute the required control inputs to achieve the state transition.

$\left. {}\right. $

\begin{theorem}
Consider that $B$ in (2) has only real eigenvalues. Then, the system (2) is
nearly controllable if and only if $B$ is nonsingular, cyclic, and does not
have a Jordan block with dimension greater than two in its Jordan canonical
form.
\end{theorem}

\begin{proof}
The proof of necessity is put in Appendix. For sufficiency, we still assume
that $B$ is of the Jordan canonical form given in (4) without loss of
generality. According to Lemma 1, there exist $\left( 2m+3\right) $ nonzero
real and pairwise distinct values $u_{0},u_{1},\ldots ,u_{2m+1},u_{2m+2}$
such that%
\begin{equation}
\left( I+u_{2m+2}B\right) \left( I+u_{2m+1}B\right) \cdots \left(
I+u_{1}B\right) \left( I+u_{0}B\right) =I\text{.}
\end{equation}%
We next prove that, for almost any $\xi $ in $%
\mathbb{R}
^{n}$, we can construct control inputs such that $\xi $ can be transferred
to an arbitrary state which is close to $\xi $. We first prove the existence
of such control inputs, which is similar to what we have done in [18]. We
then show how to compute such control inputs by applying the root locus
approach, which is concluded in the final step of Algorithm 5 (to shorten
the proof we do not write it here).

Consider the following function%
\begin{align}
F\left( t_{0},t_{1},\ldots ,t_{n-2},t_{n-1},y\right) & :=\left(
I+t_{n-1}B\right) \left( I+t_{n-2}B\right) \cdots \left( I+t_{1}B\right)
\left( I+t_{0}B\right) \xi -  \notag \\
& \left( I+u_{n-1}B\right) \left( I+u_{n-2}B\right) \cdots \left(
I+u_{1}B\right) \left( I+u_{0}B\right) \left( \xi +y\right)  \notag
\end{align}%
where $t_{0},t_{1},\ldots ,t_{n-2},t_{n-1}\in
\mathbb{R}
$ and $y\in
\mathbb{R}
^{n}$. Apparently, $\left( u_{0},u_{1},\ldots ,u_{n-2},u_{n-1},\mathbf{0}%
\right) $ is a zero of $F$. Using Lemma 1 in [18] yields%
\begin{equation}
\left\vert \frac{\partial \left( F\left( t_{0},t_{1},\ldots
,t_{n-2},t_{n-1},y\right) \right) }{\partial \left( t_{0},t_{1},\ldots
,t_{n-2},t_{n-1}\right) }\right\vert =\left\vert B\right\vert V\left(
t_{n-1},t_{n-2},\ldots ,t_{1},t_{0}\right) \left\vert
\begin{array}{ccccc}
\xi & B\xi & \cdots & B^{n-2}\xi & B^{n-1}\xi%
\end{array}%
\right\vert
\end{equation}%
where $V$ is the Vandermonde determinant and $\left\vert \cdot \right\vert $
denotes the determinant of a matrix. Thus, for any%
\begin{equation*}
\xi \notin \left\{ \xi \left\vert \text{ }\left\vert
\begin{array}{ccccc}
\xi & B\xi & \cdots & B^{n-2}\xi & B^{n-1}\xi%
\end{array}%
\right\vert =0\right. \right\} ,
\end{equation*}%
the Jacobian determinant (12) does not vanish at $\left( u_{0},u_{1},\ldots
,u_{n-2},u_{n-1},\mathbf{0}\right) $. According to the implicit function
theorem, there exist two open neighborhoods, named%
\begin{equation*}
O\left( \left[
\begin{array}{ccccc}
u_{0} & u_{1} & \cdots & u_{n-2} & u_{n-1}%
\end{array}%
\right] ^{T},\rho _{1}\right) ,\text{ }O\left( \mathbf{0},\rho _{2}\right)
\end{equation*}%
respectively, such that%
\begin{equation*}
F\left( t_{0}\left( y\right) ,t_{1}\left( y\right) ,\ldots ,t_{n-2}\left(
y\right) ,t_{n-1}\left( y\right) ,y\right) =\mathbf{0}
\end{equation*}%
where $\left[
\begin{array}{ccccc}
t_{0}\left( y\right) & t_{1}\left( y\right) & \cdots & t_{n-2}\left( y\right)
& t_{n-1}\left( y\right)%
\end{array}%
\right] ^{T}\in O\left( \left[
\begin{array}{ccccc}
u_{0} & u_{1} & \cdots & u_{n-2} & u_{n-1}%
\end{array}%
\right] ^{T},\rho _{1}\right) $ \ and $y\in O\left( \mathbf{0},\rho
_{2}\right) $, i.e.%
\begin{eqnarray*}
&&\left( I+t_{n-1}\left( y\right) B\right) \left( I+t_{n-2}\left( y\right)
B\right) \cdots \left( I+t_{1}\left( y\right) B\right) \left( I+t_{0}\left(
y\right) B\right) \xi \\
&=&\left( I+u_{n-1}B\right) \left( I+u_{n-2}B\right) \cdots \left(
I+u_{1}B\right) \left( I+u_{0}B\right) \left( \xi +y\right) .
\end{eqnarray*}%
Then, by (11) we have%
\begin{eqnarray}
&&\left( I+u_{2m+2}B\right) \cdots \left( I+u_{n}B\right) \left(
I+t_{n-1}\left( y\right) B\right) \cdots \left( I+t_{0}\left( y\right)
B\right) \xi  \notag \\
&=&\left( I+u_{2m+2}B\right) \cdots \left( I+u_{n}B\right) \left(
I+u_{n-1}B\right) \cdots \left( I+u_{0}B\right) \left( \xi +y\right) =\xi +y.
\end{eqnarray}%
This means $\xi $ can be transferred to any state that is close enough to $%
\xi $. From $B$'s structure and PBH test [22],%
\begin{equation*}
\left\{ \xi \left\vert \text{ }\left\vert
\begin{array}{ccccc}
\xi & B\xi & \cdots & B^{n-2}\xi & B^{n-1}\xi%
\end{array}%
\right\vert =0\right. \right\} =\left\{ \xi =\left[
\begin{array}{cccc}
\xi _{1} & \xi _{2} & \cdots & \xi _{n}%
\end{array}%
\right] ^{T}\left\vert \text{ }\xi _{2}\xi _{4}\cdots \xi _{2r}\xi
_{2r+1}\cdots \xi _{n}=0\right. \right\}
\end{equation*}%
that is a hypersurface in $%
\mathbb{R}
^{n}$ and separates $%
\mathbb{R}
^{n}$ into $2^{m}$ open orthants (which are the same as those in (12) in
[19]). We now prove that the system (2) is controllable on each of the $2^{m}
$ open orthants. For any two states $\xi ,\eta $ in one orthant, we
establish the \textit{transition} \textit{matrix}%
\begin{equation}
T_{\xi \rightarrow \eta }=\left[
\begin{array}{cccccccc}
\frac{\eta _{2}}{\xi _{2}} & \frac{\eta _{1}}{\xi _{2}}-\frac{\xi _{1}\eta
_{2}}{\xi _{2}^{2}} &  &  &  &  &  &  \\
&
\begin{array}{c}
\\
\frac{\eta _{2}}{\xi _{2}}%
\end{array}
&  &  &  &  &  &  \\
&  & \ddots &  &  &  &  &  \\
&  &  &
\begin{array}{c}
\\
\frac{\eta _{2r}}{\xi _{2r}}%
\end{array}
&
\begin{array}{c}
\\
\frac{\eta _{2r-1}}{\xi _{2r}}-\frac{\xi _{2r-1}\eta _{2r}}{\xi _{2r}^{2}}%
\end{array}
&  &  &  \\
&  &  &  &
\begin{array}{c}
\\
\begin{array}{c}
\frac{\eta _{2r}}{\xi _{2r}} \\
\end{array}%
\end{array}
&  &  &  \\
&  &  &  &  &
\begin{array}{c}
\\
\frac{\eta _{2r+1}}{\xi _{2r+1}} \\
\end{array}
&  &  \\
&  &  &  &  &  & \ddots &  \\
&  &  &  &  &  &  & \frac{\eta _{n}}{\xi _{n}}%
\end{array}%
\right] .
\end{equation}%
It can be seen that $\eta =T_{\xi \rightarrow \eta }\xi $ and all the
eigenvalues of $T_{\xi \rightarrow \eta }$ are positive since $\xi ,\eta $
belong to the same orthant. Furthermore,%
\begin{eqnarray*}
\underset{q\rightarrow +\infty }{\lim }T_{\xi \rightarrow \eta }^{\frac{1}{q}%
} &=&\underset{q\rightarrow +\infty }{\lim }\left[
\begin{array}{cccccccc}
\left( \frac{\eta _{2}}{\xi _{2}}\right) ^{\frac{1}{q}} & \frac{\frac{\eta
_{1}}{\xi _{2}}-\frac{\xi _{1}\eta _{2}}{\xi _{2}^{2}}}{q\left( \frac{\eta
_{2}}{\xi _{2}}\right) ^{\frac{q-1}{q}}} &  &  &  &  &  &  \\
& \left( \frac{\eta _{2}}{\xi _{2}}\right) ^{\frac{1}{q}} &  &  &  &  &  &
\\
&  & \ddots &  &  &  &  &  \\
&  &  & \left( \frac{\eta _{2r}}{\xi _{2r}}\right) ^{\frac{1}{q}} & \frac{%
\frac{\eta _{2r-1}}{\xi _{2r}}-\frac{\xi _{2r-1}\eta _{2r}}{\xi _{2r}^{2}}}{%
q\left( \frac{\eta _{2r}}{\xi _{2}r}\right) ^{\frac{q-1}{q}}} &  &  &  \\
&  &  &  & \left( \frac{\eta _{2r}}{\xi _{2r}}\right) ^{\frac{1}{q}} &  &  &
\\
&  &  &  &  & \left( \frac{\eta _{2r+1}}{\xi _{2r+1}}\right) ^{\frac{1}{q}}
&  &  \\
&  &  &  &  &  & \ddots &  \\
&  &  &  &  &  &  & \left( \frac{\eta _{n}}{\xi _{n}}\right) ^{\frac{1}{q}}%
\end{array}%
\right] \\
&=&I.
\end{eqnarray*}%
Therefore, we can choose a positive integer $q$ such that $T_{\xi
\rightarrow \eta }^{\frac{1}{q}}\xi $ is sufficiently close to $\xi $ and
hence can be reached from $\xi $\ according to (13). That is, there exist
control inputs $\bar{u}_{0},\ldots ,\bar{u}_{n-1},u_{n},\ldots ,u_{2m+2}$
such that%
\begin{equation*}
\left( I+u_{2m+2}B\right) \cdots \left( I+u_{n}B\right) \left( I+\bar{u}%
_{n-1}B\right) \cdots \left( I+\bar{u}_{0}B\right) \xi =T_{\xi \rightarrow
\eta }^{\frac{1}{q}}\xi ,
\end{equation*}%
where $\left[
\begin{array}{ccccc}
\bar{u}_{0} & \bar{u}_{1} & \cdots & \bar{u}_{n-2} & \bar{u}_{n-1}%
\end{array}%
\right] ^{T}\in O\left( \left[
\begin{array}{ccccc}
u_{0} & u_{1} & \cdots & u_{n-2} & u_{n-1}%
\end{array}%
\right] ^{T},\rho _{1}\right) $. Note that $T_{\xi \rightarrow \eta }^{\frac{%
1}{q}}$ and $B$ commute with each other. Applying $q$ groups of $\bar{u}%
_{0},\ldots ,\bar{u}_{n-1},u_{n},\ldots ,u_{2m+2}$ yields%
\begin{eqnarray*}
&&\left[ \left( I+u_{2m+2}B\right) \cdots \left( I+u_{n}B\right) \left( I+%
\bar{u}_{n-1}B\right) \cdots \left( I+\bar{u}_{0}B\right) \right] ^{q}\xi \\
&=&\left[ \left( I+u_{2m+2}B\right) \cdots \left( I+u_{n}B\right) \left( I+%
\bar{u}_{n-1}B\right) \cdots \left( I+\bar{u}_{0}B\right) \right]
^{q-1}T_{\xi \rightarrow \eta }^{\frac{1}{q}}\xi \\
&=&T_{\xi \rightarrow \eta }^{\frac{1}{q}}\left[ \left( I+u_{2m+2}B\right)
\cdots \left( I+u_{n}B\right) \left( I+\bar{u}_{n-1}B\right) \cdots \left( I+%
\bar{u}_{0}B\right) \right] ^{q-1}\xi \\
&=&\cdots =\left( T_{\xi \rightarrow \eta }^{\frac{1}{q}}\right) ^{q}\xi
=T_{\xi \rightarrow \eta }\xi =\eta .
\end{eqnarray*}%
That is, controllability on each of the $2^{m}$ open orthants has been
proved. The rest is to \textquotedblleft connect\textquotedblright\ the $%
2^{m}$ open orthants, i.e. to prove that the system is controllable on the
union of the $2^{m}$ open orthants. One can readily finish this by using
Lemma 5 in [18] and following the arguments used in [18] (pp. 2856-2857,
from eq. (33) to eq. (37)), where the only explanation we should make here
is that, although $\lambda _{1},\ldots ,\lambda _{r}$ are double eigenvalues
of $B$, the fact does affect connecting\ the $2^{m}$ open orthants. Indeed,
just by replacing the subscript $n$ by $m$ in equations from (33) to (37) in
[18], we can complete the proof.

So far, we have proved that the system (2) is controllable on%
\begin{equation}
\mathbb{R}
^{n}\setminus \left\{ \xi \left\vert \text{ }\left\vert
\begin{array}{ccccc}
\xi & B\xi & \cdots & B^{n-2}\xi & B^{n-1}\xi%
\end{array}%
\right\vert =0\right. \right\} .
\end{equation}%
Recalling the near-controllability definition, we have that%
\begin{equation*}
\mathcal{E}=\mathcal{F}=\left\{ \xi \left\vert \text{ }\left\vert
\begin{array}{ccccc}
\xi & B\xi & \cdots & B^{n-2}\xi & B^{n-1}\xi%
\end{array}%
\right\vert =0\right. \right\}
\end{equation*}%
and the system (2) is nearly controllable since $\mathcal{E},\mathcal{F}$
are two sets of zero Lebesgue measure in $%
\mathbb{R}
^{n}$.
\end{proof}

$\left. {}\right. $

\begin{remark}
Lemma 1 is important to obtain the stronger result on near-controllability
of the system (2). One can see that Lemma 2 in [18] is a special case of
Lemma 1 and Theorem 2 in [18] is a special case of Theorem 2 also. In fact,
[18] was focusing on the case of B being diagonalizable, while the case
dealt with in this paper is more general. Moreover, the \textit{transition}
matrix $T_{\xi \rightarrow \eta }$ with the root locus approach will make it
possible to compute the required control inputs that achieve the state
transition.
\end{remark}

$\left. {}\right. $

\begin{remark}
One can further prove that $\mathcal{F}$ in the proof of Theorem 2 can be $%
\varnothing $. That is, for any $\xi $ in (15) and any $\eta \in
\mathbb{R}
^{n}$, $\xi $ can be transferred to $\eta $. See the corresponding part in
the proof of Theorem 1 in [19] (from the last equation in p. 653 to the end
of the proof) for reference.
\end{remark}

$\left. {}\right. $

By using the root locus approach, an algorithm is given to compute the
required control inputs that steer the nearly controllable system (2) from
one state to another, which both belong to (15).

$\left. {}\right. $

\begin{algorithm}
Steps on computing control inputs for given initial and terminal states:
\end{algorithm}

\begin{itemize}
\item 1. Transform $B$ into the Jordan canonical form as given in (4) by a
nonsingular matrix $P$. The initial and terminal states $\xi ,\eta $ are
thus transformed into $P\xi ,P\eta $, respectively.
\end{itemize}

\begin{itemize}
\item 2. Find the control inputs that transfer $P\xi $ to a state $\zeta $
which belongs to the same orthant as $P\eta $ belongs to (Lemma 5 in [18]
will be helpful and its proof includes the details on how to find such
control inputs).

\item 3. Get the transition matrix $T_{\zeta \rightarrow P\eta }$ for $\zeta
,P\eta $ from (14).

\item 4. Choose $\lambda _{m+1},\lambda _{m+2}$ that satisfy (9).

\item 5. Choose a positive integer $q$ and compute $T_{\zeta \rightarrow
P\eta }^{\frac{1}{q}}$.

\item 6. Obtain the root loci of the characteristic equation of the
following closed loop transfer function%
\begin{equation}
1+KG\left( s\right) :=1+\frac{K\left( \left( -1\right) ^{2m+2}\mu
_{1}s^{2m+2}+\cdots +\left( -1\right) \mu _{2m}s+1\right) }{s\left(
s+\lambda _{1}\right) ^{2}\cdots \left( s+\lambda _{m}\right) ^{2}\left(
s+\lambda _{m+1}\right) \left( s+\lambda _{m+2}\right) }
\end{equation}

where%
\begin{equation}
\left[
\begin{array}{c}
\mu _{1} \\
\vdots \\
\mu _{2m}%
\end{array}%
\right] :=\left[
\begin{array}{ccc}
\lambda _{1}^{2m+2} & \cdots & \lambda _{1} \\
\left( 2m+2\right) \lambda _{1}^{2m+1} & \cdots & 1 \\
\vdots & \vdots & \vdots \\
\lambda _{m}^{2m+2} & \cdots & \lambda _{m} \\
\left( 2m+2\right) \lambda _{m}^{2m+1} & \cdots & 1 \\
\lambda _{m+1}^{2m+2} & \cdots & \lambda _{m+1} \\
\lambda _{m+2}^{2m+2} & \cdots & \lambda _{m+2}%
\end{array}%
\right] ^{-1}\left[
\begin{array}{c}
T_{\zeta \rightarrow P\eta }^{\frac{1}{q}}\left( 2,2\right) -1 \\
T_{\zeta \rightarrow P\eta }^{\frac{1}{q}}\left( 1,2\right) \\
\vdots \\
T_{\zeta \rightarrow P\eta }^{\frac{1}{q}}\left( 2r,2r\right) -1 \\
T_{\zeta \rightarrow P\eta }^{\frac{1}{q}}\left( 2r-1,2r\right) \\
T_{\zeta \rightarrow P\eta }^{\frac{1}{q}}\left( 2r+1,2r+1\right) -1 \\
0 \\
\vdots \\
T_{\zeta \rightarrow P\eta }^{\frac{1}{q}}\left( n,n\right) -1 \\
0 \\
0 \\
0%
\end{array}%
\right]
\end{equation}

with $T_{\zeta \rightarrow P\eta }^{\frac{1}{q}}\left( i,j\right) $ denoting
the $\left( i,j\right) $th entry of $T_{\zeta \rightarrow P\eta }^{\frac{1}{q%
}}$ and $K$\ increases from $0$ to $+\infty $. If any of the root loci
leaves the real axis directly at the pole, return to the former step and
choose another integer $q$ greater than the previous one. Otherwise, choose
a suitable $K$ such that the roots of $1+KG\left( s\right) =0$ are all real.
Then, the reciprocals of the real roots are the control inputs that transfer
$\zeta $ to $T_{\zeta \rightarrow P\eta }^{\frac{1}{q}}\zeta $. $q$ groups
of such control inputs together with the control inputs that transfer $P\xi $
to $\zeta $ are the desired ones which steer the nearly controllable system
(2) from $\xi $ to $\eta $.
\end{itemize}

$\left. {}\right. $

\emph{An Explanation to Step 6 of Algorithm 5.} Consider the equation%
\begin{equation}
\left( I+v_{2m+2}B\right) \left( I+v_{2m+1}B\right) \cdots \left(
I+v_{1}B\right) \left( I+v_{0}B\right) \zeta =T_{\zeta \rightarrow P\eta }^{%
\frac{1}{q}}\zeta .
\end{equation}%
From the proof of Theorem 2, we know that if $q$ is large enough, then eq.
(18) admits a real solution of $v_{0},v_{1},\ldots ,v_{2m+1},v_{2m+2}$. Eq.
(18) is equivalent to the matrix equation%
\begin{equation*}
\left( I+v_{2m+2}B\right) \left( I+v_{2m+1}B\right) \cdots \left(
I+v_{1}B\right) \left( I+v_{0}B\right) =T_{\zeta \rightarrow P\eta }^{\frac{1%
}{q}}.
\end{equation*}%
Let%
\begin{equation}
\overset{2m+2}{\underset{k=0}{\prod }}\frac{1}{v_{k}}=\left( -1\right)
^{2m+3}K.
\end{equation}%
As shown for deriving eq. (7) in the proof of Lemma 1, we can obtain%
\begin{equation*}
B^{2m+3}+\overset{2m+2}{\underset{k=0}{\sum }}\frac{1}{v_{k}}B^{2m+2}+\cdots
+\overset{2m+2}{\underset{k=0}{\sum }}\frac{v_{k}}{\overset{2m+2}{\underset{%
j=0}{\prod }}v_{j}}B=\left( -1\right) ^{2m+3}K\left( T_{\zeta \rightarrow
P\eta }^{\frac{1}{q}}-I\right)
\end{equation*}%
which can be written via equations as those in (8). Using $\lambda
_{m+1},\lambda _{m+2}$ chosen in Step 4 and introducing the same constraints
as those for (8), we can deduce%
\begin{multline*}
\left[
\begin{array}{ccc}
\lambda _{1}^{2m+2} & \cdots & \lambda _{1} \\
\left( 2m+2\right) \lambda _{1}^{2m+1} & \cdots & 1 \\
\vdots & \vdots & \vdots \\
\lambda _{m}^{2m+2} & \cdots & \lambda _{m} \\
\left( 2m+2\right) \lambda _{m}^{2m+1} & \cdots & 1 \\
\lambda _{m+1}^{2m+2} & \cdots & \lambda _{m+1} \\
\lambda _{m+2}^{2m+2} & \cdots & \lambda _{m+2}%
\end{array}%
\right] \left[
\begin{array}{c}
\overset{2m+2}{\underset{k=0}{\sum }}\frac{1}{v_{k}} \\
\overset{}{\underset{}{\vdots }} \\
\overset{2m+2}{\underset{k=0}{\sum }}\frac{v_{k}}{\overset{2m+2}{\underset{%
j=0}{\prod }}v_{j}}%
\end{array}%
\right] \\
=\left[
\begin{array}{c}
-\lambda _{1}^{2m+3} \\
-\left( 2m+3\right) \lambda _{1}^{2m+2} \\
\vdots \\
-\lambda _{m}^{2m+3} \\
-\left( 2m+3\right) \lambda _{m}^{2m+2} \\
-\lambda _{m+1}^{2m+3} \\
-\lambda _{m+2}^{2m+3}%
\end{array}%
\right] +\left( -1\right) ^{2m+3}K\left[
\begin{array}{c}
T_{\zeta \rightarrow P\eta }^{\frac{1}{q}}\left( 2,2\right) -1 \\
T_{\zeta \rightarrow P\eta }^{\frac{1}{q}}\left( 1,2\right) \\
\vdots \\
T_{\zeta \rightarrow P\eta }^{\frac{1}{q}}\left( 2r,2r\right) -1 \\
T_{\zeta \rightarrow P\eta }^{\frac{1}{q}}\left( 2r-1,2r\right) \\
T_{\zeta \rightarrow P\eta }^{\frac{1}{q}}\left( 2r+1,2r+1\right) -1 \\
0 \\
\vdots \\
T_{\zeta \rightarrow P\eta }^{\frac{1}{q}}\left( n,n\right) -1 \\
0 \\
0 \\
0%
\end{array}%
\right] .
\end{multline*}%
By Lemma 15,%
\begin{equation*}
\left[
\begin{array}{c}
\overset{2m+2}{\underset{k=0}{\sum }}\frac{1}{v_{k}} \\
\overset{}{\underset{}{\vdots }} \\
\overset{2m+2}{\underset{k=0}{\sum }}\frac{v_{k}}{\overset{2m+2}{\underset{%
j=0}{\prod }}v_{j}}%
\end{array}%
\right] =\left[
\begin{array}{c}
\left( -1\right) \left( 2\overset{m}{\underset{i=1}{\sum }}\lambda
_{i}+\lambda _{m+1}+\lambda _{m+2}\right) \\
\overset{}{\underset{}{\vdots }} \\
\left( -1\right) ^{2m+2}\lambda _{m+1}\lambda _{m+2}\overset{m}{\underset{i=1%
}{\prod }}\lambda _{i}^{2}%
\end{array}%
\right] +\left( -1\right) ^{2m+3}K\left[
\begin{array}{c}
\mu _{1} \\
\vdots \\
\mu _{2m}%
\end{array}%
\right]
\end{equation*}%
where $\mu _{1},\ldots ,\mu _{2m}$ are given in (17). Then, with (19) we
have that $\frac{1}{v_{0}},\frac{1}{v_{1}},\ldots ,\frac{1}{v_{2m+1}},\frac{1%
}{v_{2m+2}}$ are the roots of the following$\ \left( 2m+3\right) $th-degree
equation%
\begin{multline*}
s^{2m+3}+\left( 2\overset{m}{\underset{i=1}{\sum }}\lambda _{i}+\lambda
_{m+1}+\lambda _{m+2}+\left( -1\right) ^{2m+2}K\mu _{1}\right)
s^{2m+2}+\cdots + \\
\left( \lambda _{m+1}\lambda _{m+2}\overset{m}{\underset{i=1}{\prod }}%
\lambda _{i}^{2}+\left( -1\right) K\mu _{2m}\right) s+K=0
\end{multline*}%
which is equivalent to%
\begin{multline*}
s\left( s+\lambda _{1}\right) ^{2}\cdots \left( s+\lambda _{m}\right)
^{2}\left( s+\lambda _{m+1}\right) \left( s+\lambda _{m+2}\right) + \\
K\left( \left( -1\right) ^{2m+2}\mu _{1}s^{2m+2}+\cdots +\left( -1\right)
\mu _{2m}s+1\right) =0.
\end{multline*}%
Thus, $\frac{1}{v_{0}},\frac{1}{v_{1}},\ldots ,\frac{1}{v_{2m+1}},\frac{1}{%
v_{2m+2}}$ are on the root loci of the characteristic equation of the closed
loop transfer function given in (16).

$\left. {}\right. $

An example will be provided in Section 4 to show the effectiveness of
Algorithm 5.

$\left. {}\right. $

\begin{remark}
By using the similar idea, one can try to derive an algorithm for computing
the control inputs to achieve the state transition for the nearly
controllable bilinear systems studied in [19]. Furthermore, although the
result obtained in [19], namely Theorem 1 in [19], seems similar to Theorem
2, it cannot yield Theorem 2 (vice versa). To see this, for $\xi ,\eta $ of
the system (2) that belong to (15), from Theorem 1 in [19] we can have $%
u(0),u(1),\ldots ,u(L_{1}-1),u\left( L_{1}\right) $ such that%
\begin{equation*}
\left( B^{-1}+u\left( L_{1}\right) I\right) \left( B^{-1}+u\left(
L_{1}-1\right) I\right) \cdots \left( B^{-1}+u\left( 1\right) I\right)
\left( B^{-1}+u\left( 0\right) I\right) \xi =B^{-L_{2}-1}\eta
\end{equation*}%
where $B$ satisfies the conditions in Theorem 2 and $\xi ,$ $%
B^{-L_{2}-1}\eta $ are the corresponding initial and terminal states,
respectively. This implies%
\begin{equation*}
\left( I+u\left( L_{1}\right) B\right) \left( I+u\left( L_{1}-1\right)
B\right) \cdots \left( I+u\left( 1\right) B\right) \left( I+u\left( 0\right)
B\right) \xi =B^{L_{1}-L_{2}}\eta .
\end{equation*}%
Unfortunately, we can transfer $\xi $ to $B^{L_{1}-L_{2}}\eta $ but not to $%
\eta $ since we do not have $L_{1}=L_{2}$ either from Theorem 1 in [19] or
from Theorem 2.
\end{remark}

$\left. {}\right. $

\begin{remark}
Note that the opposite numbers of $B$'s eigenvalues are poles of $G\left(
s\right) $ given in (16), which are all real so that it is possible for the
root loci of $1+KG\left( s\right) =0$ to first move on the real axis.
However, if $B$ has complex eigenvalues, then $G\left( s\right) $ may have
complex poles and some of the root loci of $1+KG\left( s\right) =0$ will not
start at the real axis. In such a case, it is rather difficult to ensure
that all the root loci can have the same moment moving on the real axis,
then the real control inputs that achieve the state transition cannot be
obtained. Therefore, for near-controllability of system (2) with $B$ having
complex eigenvalues, we need to develop the proposed root locus approach or
to find a new method.
\end{remark}

$\left. {}\right. $

\section{Nearly-controllable subspaces and Near-controllability index}

Consider the system%
\begin{equation}
x\left( k+1\right) =\left( I+u\left( k\right) B\right) x\left( k\right)
=\left( I+u\left( k\right) \left[
\begin{array}{ccc}
\lambda & 1 & 0 \\
0 & \lambda & 1 \\
0 & 0 & \lambda%
\end{array}%
\right] \right) x\left( k\right)
\end{equation}%
where $x\left( k\right) \in
\mathbb{R}
^{3}$, $u\left( k\right) \in
\mathbb{R}
$, and $\lambda \neq 0$. Since $B$ has a Jordan block with dimension greater
than two, system (20) is non-nearly controllable according to Theorem 2.
Nevertheless, consider system (20) on region%
\begin{equation}
\left\{ \xi =\left[
\begin{array}{ccc}
\xi _{1} & \xi _{2} & 0%
\end{array}%
\right] ^{T}\right\} .
\end{equation}%
It can be seen that the system is invariant on (21), i.e.%
\begin{eqnarray}
x_{1,2}(k+1) &=&\left( I+u\left( k\right) B_{1,2}\right) x_{1,2}(k)=\left(
I+u\left( k\right) \left[
\begin{array}{cc}
\lambda & 1 \\
0 & \lambda%
\end{array}%
\right] \right) x_{1,2}(k), \\
x_{3}(k+1) &=&0  \notag
\end{eqnarray}%
where $x_{1,2}(k)=\left[
\begin{array}{cc}
x_{1}(k) & x_{2}(k)%
\end{array}%
\right] ^{T}$ and $B_{1,2}$\ is the main submatrix of $B$\ by taking the
entries of $B$ in both rows $1,2$\ and columns $1,2$. From Theorem 2,
subsystem (22) is nearly controllable. More specifically, it is controllable
on%
\begin{equation}
\mathbb{R}
^{2}\setminus \left\{ \xi _{1,2}\left\vert \text{ }\left\vert
\begin{array}{cc}
\xi _{1,2} & B_{1,2}\xi _{1,2}%
\end{array}%
\right\vert =0\right. \right\} =\left\{ \xi _{1,2}=\left[
\begin{array}{cc}
\xi _{1} & \xi _{2}%
\end{array}%
\right] ^{T}\left\vert \text{ }\xi _{2}\neq 0\right. \right\} .
\end{equation}%
Therefore, system (20) is controllable on%
\begin{equation}
\left\{ \xi _{1,2}=\left[
\begin{array}{cc}
\xi _{1} & \xi _{2}%
\end{array}%
\right] ^{T}\left\vert \text{ }\xi _{2}\neq 0\right. \right\} \otimes
\left\{ \xi _{3}=0\right\} =\left\{ \xi =\left[
\begin{array}{ccc}
\xi _{1} & \xi _{2} & 0%
\end{array}%
\right] ^{T}\left\vert \text{ }\xi _{2}\neq 0\right. \right\}
\end{equation}%
which is a two-dimensional region in $%
\mathbb{R}
^{3}$. Similarly, one can deduce that the one-dimensional region%
\begin{equation*}
\left\{ \xi =\left[
\begin{array}{ccc}
\xi _{1} & 0 & 0%
\end{array}%
\right] ^{T}\left\vert \text{ }\xi _{1}\neq 0\right. \right\}
\end{equation*}%
is also a region on which system (20) is controllable.

$\left. {}\right. $

In this section, we study the non-nearly controllable systems. It is well
known that if a linear time-invariant system is uncontrollable, then the
state space can be decomposed as a direct sum of a controllable subspace and
an uncontrollable subspace. On the controllable subspace the linear system
is controllable. Furthermore, the controllable subspace corresponds to a
controllable linear subsystem. Inspired by these facts, we will derive
nearly-controllable subspaces of the system (2) by using Theorem 2.
Actually, the regions in (23) and (24) are nearly-controllable subspaces of
system (20). We do not use \textquotedblleft controllable
subspace\textquotedblright\ here since a subspace is in general a linear
space, while a nearly-controllable subspace does not contain the zero
element (the origin is an isolated state of system (1) that once it is
reached, then the system cannot be steered away).

$\left. {}\right. $

\begin{definition}
A nearly-controllable subspace of the system (2) is a controllable region of
the system (2) that is derived from the corresponding nearly controllable
subsystem of the system (2) on the region.
\end{definition}

$\left. {}\right. $

By the following Lemma, we will show how to obtain a nearly-controllable
subspace.

$\left. {}\right. $

\begin{lemma}
Consider the system%
\begin{equation}
x\left( k+1\right) =\left( I+u\left( k\right) B\right) x\left( k\right)
=\left( I+u\left( k\right) \left[
\begin{array}{cccc}
\lambda & 1 &  &  \\
& \lambda & \ddots &  \\
&  & \ddots & 1 \\
&  &  & \lambda%
\end{array}%
\right] \right) x\left( k\right)
\end{equation}%
where $x\left( k\right) \in
\mathbb{R}
^{n}$, $n\geq 3$, $u\left( k\right) \in
\mathbb{R}
$, and $\lambda \neq 0$. Then, any state in a controllable region of system
(25), which contains more than one state, has the property that only the
first two entries can be nonzero.
\end{lemma}

$\left. {}\right. $

The proof of this lemma can be found in Appendix.

$\left. {}\right. $

\begin{theorem}
Consider the system (2) with $B$ in Jordan canonical form%
\begin{equation*}
\left[
\begin{array}{ccc}
J\left( \lambda _{1}\right) &  &  \\
& \ddots &  \\
&  & J\left( \lambda _{m}\right)%
\end{array}%
\right]
\end{equation*}%
without loss of generality, where $\lambda _{i}\neq \lambda _{j}$ if $i\neq
j $ and $J\left( \lambda _{i}\right) $ is the Jordan matrix associated with
eigenvalue $\lambda _{i}$ for $i=1,\ldots ,m$. Let $h$ be the dimension of $%
B $'s largest main submatrix that is nonsingular, cyclic, and does not have
a Jordan block with dimension greater than two. Then, the system (2) has an $%
e$-dimensional nearly-controllable subspace for $e=1,\ldots ,h$.
\end{theorem}

\begin{proof}
Choose any main submatrix $B_{i_{1},\ldots ,i_{e}}$ of $B$ that is
nonsingular, cyclic, and does not have a Jordan block with dimension greater
than two, where $i_{j}$ is required to correspond to the first or the second
row (and column) of a Jordan block in $B$ in view of Lemma 9 for $j=1,\ldots
,e$ and if some $i_{j}$ corresponds to the second row (and column) of a
Jordan block, then $i_{j-1}$ corresponds to the first row (and column) of
the same one. Consider the system (2) on%
\begin{equation*}
\left\{ \xi =\left[
\begin{array}{ccccccccccc}
0 & \cdots & 0 & \xi _{i_{1}} & 0 & \cdots & 0 & \xi _{i_{e}} & 0 & \cdots &
0%
\end{array}%
\right] ^{T}\right\} .
\end{equation*}%
Then, the system can be rewritten as%
\begin{eqnarray}
x_{i_{1},\ldots ,i_{e}}(k+1) &=&\left( I+u\left( k\right) B_{i_{1},\ldots
,i_{e}}\right) x_{i_{1},\ldots ,i_{e}}(k), \\
x_{1,\ldots ,i_{1}-1,i_{1}+1,\ldots ,i_{e}-1,i_{e}+1,\ldots ,n}(k+1) &=&%
\mathbf{0}  \notag
\end{eqnarray}%
where%
\begin{equation*}
x_{i_{1},\ldots ,i_{e}}(k)=\left[
\begin{array}{cccc}
x_{i_{1}}(k) & x_{i_{2}}(k) & \cdots & x_{i_{e}}(k)%
\end{array}%
\right] ^{T}
\end{equation*}%
and%
\begin{multline*}
x_{1,\ldots ,i_{1}-1,i_{1}+1,\ldots ,i_{e}-1,i_{e}+1,\ldots ,n}(k) \\
=\left[
\begin{array}{ccccccccc}
x_{1}(k) & \cdots & x_{i_{1}-1}(k) & x_{i_{1}+1}(k) & \cdots & x_{i_{e}-1}(k)
& x_{i_{e}+1}(k) & \cdots & x_{n}(k)%
\end{array}%
\right] ^{T}.
\end{multline*}%
Due to the chosen $B_{i_{1},\ldots ,i_{e}}$, subsystem (26) is nearly
controllable by Theorem 2. Thus, the system (2) is controllable on the $e $%
-dimensional region%
\begin{equation}
\left\{ \xi =\left[
\begin{array}{ccccccccccc}
0 & \cdots & 0 & \xi _{i_{1}} & 0 & \cdots & 0 & \xi _{i_{e}} & 0 & \cdots &
0%
\end{array}%
\right] ^{T}\right\} \setminus H
\end{equation}%
where%
\begin{multline*}
H=\left\{ \left. \xi =\left[
\begin{array}{ccccccccccc}
0 & \cdots & 0 & \xi _{i_{1}} & 0 & \cdots & 0 & \xi _{i_{e}} & 0 & \cdots &
0%
\end{array}%
\right] ^{T}\right\vert \right. \\
\left. \left\vert
\begin{array}{cccc}
\xi _{i_{1},\ldots ,i_{e}} & B_{i_{1},\ldots ,i_{e}}\xi _{i_{1},\ldots
,i_{e}} & \cdots & B_{i_{1},\ldots ,i_{e}}^{e-1}\xi _{i_{1},\ldots ,i_{e}}%
\end{array}%
\right\vert =0,\text{ }\xi _{i_{1},\ldots ,i_{e}}=\left[
\begin{array}{ccc}
\xi _{i_{1}} & \cdots & \xi _{i_{e}}%
\end{array}%
\right] ^{T}\right\} .
\end{multline*}%
That is, region (27) is an $e$-dimensional nearly-controllable subspace of
the system (2).
\end{proof}

$\left. {}\right. $

\begin{remark}
From the proof of Theorem 10, in order to obtain a nearly-controllable
subspace of the system (2), the main submatrix $B_{i_{1},\ldots ,i_{e}}$ of $%
B$ that the nearly-controllable subspace corresponds to must be chosen to be
nonsingular, cyclic, and have no Jordan block with dimension greater than
two, where $i_{j}$ is required to correspond to the first or the second row
(and column) of a Jordan block in $B$ for $j=1,\ldots ,e$ and if some $i_{j}$
corresponds to the second row (and column) of a Jordan block, then $i_{j-1}$
corresponds to the first row (and column) of the same one.
\end{remark}

$\left. {}\right. $

\begin{definition}
$h$\ in Theorem 10 is called the near-controllability index of the system
(2).
\end{definition}

$\left. {}\right. $

If $h=n$, then the system (2) is nearly controllable. Otherwise, we can have
the nearly-controllable subspaces with dimension from $1$ to $h$. In
particular, the state transition on every nearly-controllable subspace can
also be achieved through Algorithm 5. Additionally, even for the nearly
controllable system (2), it has nearly-controllable subspaces in the removed
region of (15) (i.e. $\mathcal{E}$) and (15) can be regarded as an $n$%
-dimensional nearly-controllable subspace. An example is given in the next
section.

$\left. {}\right. $

\section{Examples}

\begin{example}
Consider the system%
\begin{equation}
x\left( k+1\right) =\left( I+u\left( k\right) B\right) x\left( k\right)
=\left( I+u\left( k\right) \left[
\begin{array}{ccc}
2 & 1 & -5 \\
0 & 2 & -4 \\
0 & 0 & -2%
\end{array}%
\right] \right) x\left( k\right)
\end{equation}%
where $x\left( k\right) \in
\mathbb{R}
^{3}$ and $u\left( k\right) \in
\mathbb{R}
$. Given $\xi =\left[
\begin{array}{ccc}
0 & 0 & -1%
\end{array}%
\right] ^{T},$ $\eta =\left[
\begin{array}{ccc}
27 & 21 & 12%
\end{array}%
\right] ^{T}$. Find the control inputs such that $\xi $ is transferred to $%
\eta $.
\end{example}

$\left. {}\right. $

We now apply Algorithm 5 to compute the required control inputs. \textbf{%
Step 1}: let $\bar{x}\left( k\right) =Px\left( k\right) $ where%
\begin{equation*}
P=\left[
\begin{array}{ccc}
1 & -1 & 0 \\
0 & 1 & -1 \\
0 & 0 & 1%
\end{array}%
\right] .
\end{equation*}%
Then,%
\begin{equation}
\bar{x}\left( k+1\right) =\left( I+u\left( k\right) PBP^{-1}\right) \bar{x}%
\left( k\right) =\left( I+u\left( k\right) \left[
\begin{array}{ccc}
2 & 1 & 0 \\
0 & 2 & 0 \\
0 & 0 & -2%
\end{array}%
\right] \right) \bar{x}\left( k\right)
\end{equation}%
and $P\xi =\left[
\begin{array}{ccc}
0 & 1 & -1%
\end{array}%
\right] ^{T},$ $P\eta =\left[
\begin{array}{ccc}
6 & 9 & 12%
\end{array}%
\right] ^{T}$. From Theorem 2, system (29) is nearly controllable (so is
system (28)) and is controllable on%
\begin{eqnarray*}
&&%
\mathbb{R}
^{3}\setminus \left\{ \xi \left\vert \text{ }\left\vert
\begin{array}{ccc}
\xi & PBP^{-1}\xi & \left( PBP^{-1}\right) ^{2}\xi%
\end{array}%
\right\vert =0\right. \right\} \\
&=&%
\mathbb{R}
^{3}\setminus \left\{ \xi =\left[
\begin{array}{ccc}
\xi _{1} & \xi _{2} & \xi _{3}%
\end{array}%
\right] ^{T}\left\vert \text{ }\xi _{2}\xi _{3}=0\right. \right\} \\
&=&\left\{ \xi \left\vert \text{ }\xi _{2}>0,\xi _{3}>0\right. \right\} \cup
\left\{ \xi \left\vert \text{ }\xi _{2}<0,\xi _{3}>0\right. \right\} \cup
\left\{ \xi \left\vert \text{ }\xi _{2}<0,\xi _{3}<0\right. \right\} \cup
\left\{ \xi \left\vert \text{ }\xi _{2}>0,\xi _{3}<0\right. \right\}
\end{eqnarray*}%
that consists of four open orthants. \textbf{Step 2}: since $P\xi ,P\eta $
belong to different orthants, let $u\left( 0\right) =1$. It follows%
\begin{equation*}
\left( I+u\left( 0\right) PBP^{-1}\right) P\xi =\left[
\begin{array}{ccc}
1 & 3 & 1%
\end{array}%
\right] ^{T}\triangleq \zeta
\end{equation*}%
which is in the orthant that $P\eta $ belongs to. \textbf{Step 3}: obtain
from (14) the transition matrix%
\begin{equation*}
T_{\zeta \rightarrow P\eta }=\left[
\begin{array}{ccc}
3 & 1 & 0 \\
0 & 3 & 0 \\
0 & 0 & 12%
\end{array}%
\right] .
\end{equation*}%
\textbf{Step 4}: choose $\lambda _{3}=1,\lambda _{4}=-4$ in view of (9).
\textbf{Step 5}: choose $q=4$. Then,%
\begin{equation*}
T_{\zeta \rightarrow P\eta }^{\frac{1}{4}}=\left[
\begin{array}{ccc}
3^{\frac{1}{4}} & \frac{1}{4\times 3^{\frac{3}{4}}} & 0 \\
0 & 3^{\frac{1}{4}} & 0 \\
0 & 0 & 12^{\frac{1}{4}}%
\end{array}%
\right] .
\end{equation*}%
\textbf{Step 6}: by (17)%
\begin{equation*}
\mu _{1}\approx 0.000742,\text{ }\mu _{2}\approx -0.00511,\text{ }\mu
_{3}\approx -0.0393,\text{ }\mu _{4}\approx -0.0648,\text{ }\mu _{5}\approx
0.293,\text{ }\mu _{6}\approx 0.314.
\end{equation*}%
Consider%
\begin{equation*}
1+KG\left( s\right) =1+\frac{K\left( \left( -1\right) ^{6}\mu
_{1}s^{6}+\cdots +\left( -1\right) \mu _{6}s+1\right) }{s\left( s+2\right)
^{2}\left( s-2\right) ^{2}\left( s+1\right) \left( s-4\right) }.
\end{equation*}%
By Matlab the root loci of $1+KG\left( s\right) =0$ are shown in the
following figure,
\begin{figure}[h!]
\begin{center}
\includegraphics[scale=0.808,trim=0 0 0 0]{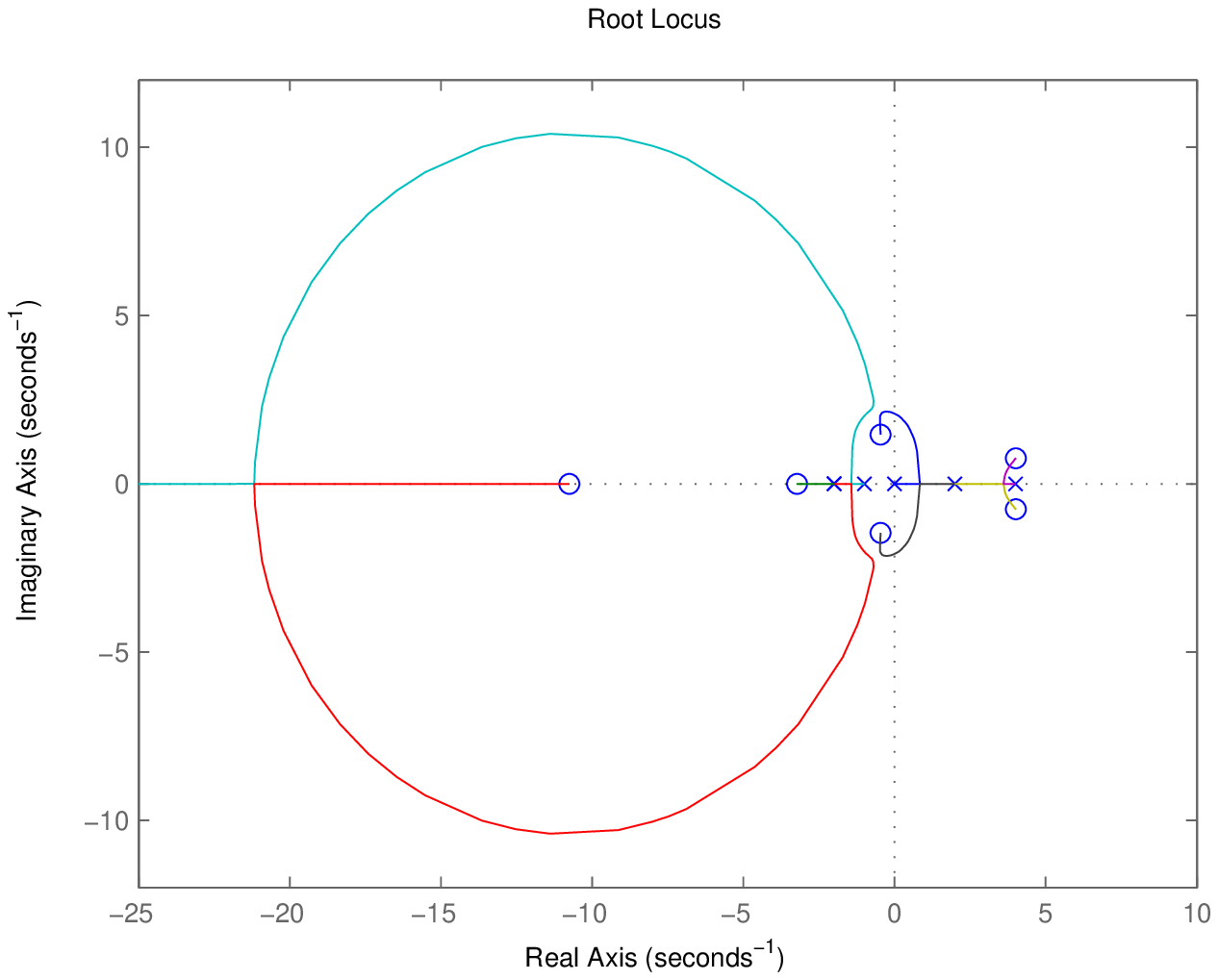} \\[0pt]
{\fontsize{11pt}{11.6pt} \selectfont Fig. 1. Root loci of $1+KG\left(
s\right) =0$}
\end{center}
\end{figure}
where \textquotedblleft $\times $\textquotedblright\ and \textquotedblleft $%
\circ $\textquotedblright\ respectively denote the poles and zeros of $%
G\left( s\right) $ and the colored curves are the root loci of $1+KG\left(
s\right) =0$ starting at the poles and ending at the zeros. From Fig. 1, the
root loci of $1+KG\left( s\right) =0$ first move on the real axis. Choose $%
K=10$ and obtain the roots of $1+KG\left( s\right) =0$. Computing the
reciprocals of the roots yields%
\begin{equation*}
v_{1}\approx -0.770,\text{ }v_{2}\approx -0.643,\text{ }v_{3}\approx -0.452,%
\text{ }v_{4}\approx 0.250,\text{ }v_{5}\approx 0.439,\text{ }v_{6}\approx
0.612,\text{ }v_{7}\approx 6.650.
\end{equation*}

One can now verify that, by $4$ groups of the above control inputs with $%
u\left( 0\right) =1$, system (28) is steered from $\xi $ to $\eta $.

$\left. {}\right. $

\begin{example}
Consider the system%
\begin{equation}
x\left( k+1\right) =\left( I+u\left( k\right) B\right) x\left( k\right)
=\left( I+u\left( k\right) \left[
\begin{array}{cccc}
\lambda _{1} & 1 & 0 & 0 \\
0 & \lambda _{1} & 1 & 0 \\
0 & 0 & \lambda _{1} & 0 \\
0 & 0 & 0 & \lambda _{2}%
\end{array}%
\right] \right) x\left( k\right)
\end{equation}%
where $x\left( k\right) \in
\mathbb{R}
^{4}$, $u(k)\in
\mathbb{R}
$, and $\lambda _{1},\lambda _{2}$ are nonzero real and distinct. Find its
nearly-controllable subspaces.
\end{example}

$\left. {}\right. $

According to Theorem 2, system (30) is non-nearly controllable. The
dimension of $B$'s largest main submatrix,\ which is nonsingular, cyclic,
and does not have a Jordan block with dimension greater than two, is three,
so the near-controllability index of system (30) $h=3$. More specifically, $%
B_{1,2,4},B_{2,3,4}$ are such main submatrices of $B$. From the proof of
Theorem 10 $B_{1,2,4}$ corresponds to the following three-dimensional
nearly-controllable subspace of system (30)%
\begin{eqnarray}
&&\left\{ \xi =\left[
\begin{array}{cccc}
\xi _{1} & \xi _{2} & 0 & \xi _{4}%
\end{array}%
\right] ^{T}\right\} \setminus \left\{ \xi \left\vert \text{ }\left\vert
\begin{array}{ccc}
\xi _{1,2,4} & B_{1,2,4}\xi _{1,2,4} & B_{1,2,4}^{2}\xi _{1,2,4}%
\end{array}%
\right\vert =0\right. \right\}  \notag \\
&=&\left\{ \xi =\left[
\begin{array}{cccc}
\xi _{1} & \xi _{2} & 0 & \xi _{4}%
\end{array}%
\right] ^{T}\left\vert \text{ }\xi _{2}\xi _{4}\neq 0\right. \right\}
\end{eqnarray}%
where $\xi _{1,2,4}=\left[
\begin{array}{ccc}
\xi _{1} & \xi _{2} & \xi _{4}%
\end{array}%
\right] ^{T}$. However, the region%
\begin{equation*}
\left\{ \xi =\left[
\begin{array}{cccc}
0 & \xi _{2} & \xi _{3} & \xi _{4}%
\end{array}%
\right] ^{T}\left\vert \text{ }\xi _{3}\xi _{4}\neq 0\right. \right\}
\end{equation*}%
that $B_{2,3,4}$ corresponds to is not a nearly-controllable subspace in
view of Lemma 9. Besides (31), system (30) has the two-dimensional
nearly-controllable subspaces%
\begin{equation*}
\left\{ \xi =\left[
\begin{array}{cccc}
\xi _{1} & \xi _{2} & 0 & 0%
\end{array}%
\right] ^{T}\text{ }\left\vert \text{ }\xi _{2}\neq 0\right. \right\} ,\text{
}\left\{ \xi =\left[
\begin{array}{cccc}
\xi _{1} & 0 & 0 & \xi _{4}%
\end{array}%
\right] ^{T}\text{ }\left\vert \text{ }\xi _{1}\xi _{4}\neq 0\right. \right\}
\end{equation*}%
that $B_{1,2},B_{1,4}$ respectively correspond to and the one-dimensional
nearly-controllable subspaces%
\begin{equation*}
\left\{ \xi =\left[
\begin{array}{cccc}
\xi _{1} & 0 & 0 & 0%
\end{array}%
\right] ^{T}\text{ }\left\vert \text{ }\xi _{1}\neq 0\right. \right\} ,\text{
}\left\{ \xi =\left[
\begin{array}{cccc}
0 & 0 & 0 & \xi _{4}%
\end{array}%
\right] ^{T}\text{ }\left\vert \text{ }\xi _{4}\neq 0\right. \right\}
\end{equation*}%
that $B_{1},B_{4}$ respectively\ corresponds to. In all, system (30) has
five nearly-controllable subspaces from Theorem 10.

$\left. {}\right. $

\section{Conclusions}

In this paper, near-controllability of a class of discrete-time bilinear
systems is studied by proposing a root locus approach. A necessary and
sufficient criterion for the systems to be nearly controllable is obtained.
In particular, the control inputs which achieve the state transition for the
nearly controllable systems are computed and an algorithm is presented.
Furthermore, the nearly-controllable subspaces are derived and the
near-controllability index is defined. Two examples are provided to
demonstrate the results of the paper. Future work should consider the
near-controllability and controllability problems of general bilinear
systems and develop the proposed root locus approach.

$\left. {}\right. $

$\left. {}\right. $

\section{Appendix}

$\left. {}\right. $

\begin{lemma}
Let%
\begin{equation*}
C=\left[
\begin{array}{ccc}
\lambda _{1}^{2m+2} & \cdots & \lambda _{1} \\
\left( 2m+2\right) \lambda _{1}^{2m+1} & \cdots & 1 \\
\vdots & \vdots & \vdots \\
\lambda _{m}^{2m+2} & \cdots & \lambda _{m} \\
\left( 2m+2\right) \lambda _{m}^{2m+1} & \cdots & 1 \\
\lambda _{m+1}^{2m+2} & \cdots & \lambda _{m+1} \\
\lambda _{m+2}^{2m+2} & \cdots & \lambda _{m+2}%
\end{array}%
\right] ,\text{ }d=\left[
\begin{array}{c}
-\lambda _{1}^{2m+3} \\
-\left( 2m+3\right) \lambda _{1}^{2m+2} \\
\vdots \\
-\lambda _{m}^{2m+3} \\
-\left( 2m+3\right) \lambda _{m}^{2m+2} \\
-\lambda _{m+1}^{2m+3} \\
-\lambda _{m+2}^{2m+3}%
\end{array}%
\right]
\end{equation*}%
where $\lambda _{1},\ldots ,\lambda _{m+2}$ are nonzero real and pairwise
distinct. Then, the linear equation%
\begin{equation}
Cz=d
\end{equation}%
has a unique solution%
\begin{equation}
z=C^{-1}d=\left[
\begin{array}{c}
\left( -1\right) \left( 2\overset{m}{\underset{i=1}{\sum }}\lambda
_{i}+\lambda _{m+1}+\lambda _{m+2}\right) \\
\overset{}{\underset{}{\vdots }} \\
\left( -1\right) ^{2m+2}\lambda _{m+1}\lambda _{m+2}\overset{m}{\underset{i=1%
}{\prod }}\lambda _{i}^{2}%
\end{array}%
\right] .  \notag
\end{equation}
\end{lemma}

\begin{proof}
Note that the transpose of\textit{\ }$C$ is a general Vandermonde\textit{\ }%
matrix. $C$ is nonsingular since $\lambda _{1},\ldots ,$ $\lambda _{m+2}$
are nonzero and pairwise distinct. Therefore, linear equation (32) has a
unique solution. Assume that $z=\left[
\begin{array}{cccc}
z_{1} & z_{2} & \cdots & z_{2m+2}%
\end{array}%
\right] ^{T}$ is the solution. From (32) we have%
\begin{equation*}
\lambda _{i}^{2m+3}+z_{1}\lambda _{i}^{2m+2}+\cdots +z_{2m+2}\lambda _{i}=0
\end{equation*}%
for $i=1,\ldots ,m+2$ and%
\begin{equation*}
\left( 2m+3\right) \lambda _{i}^{2m+2}+z_{1}\left( 2m+2\right) \lambda
_{i}^{2m+1}+\cdots +z_{2m+2}=0
\end{equation*}%
for $i=1,\ldots ,m$. Since $\lambda _{i}$ are nonzero, it follows%
\begin{equation*}
\lambda _{i}^{2m+2}+z_{1}\lambda _{i}^{2m+1}+\cdots +z_{2m+2}=0
\end{equation*}%
for $i=1,\ldots ,m+2$ and hence%
\begin{eqnarray*}
&&\left( 2m+3\right) \lambda _{i}^{2m+2}+z_{1}\left( 2m+2\right) \lambda
_{i}^{2m+1}+\cdots +z_{2m+2}-\left( \lambda _{i}^{2m+2}+z_{1}\lambda
_{i}^{2m+1}+\cdots +z_{2m+2}\right) \\
&=&\left( 2m+2\right) \lambda _{i}^{2m+2}+z_{1}\left( 2m+1\right) \lambda
_{i}^{2m+1}+\cdots +z_{2m+1}\lambda _{i}=0 \\
&\Rightarrow &\left( 2m+2\right) \lambda _{i}^{2m+1}+z_{1}\left( 2m+1\right)
\lambda _{i}^{2m}+\cdots +z_{2m+1}=0
\end{eqnarray*}%
for $i=1,\ldots ,m$. We can thus conclude that $\lambda _{1},\ldots ,\lambda
_{m+2}$ are all the roots of the following $\left( 2m+2\right) $th-degree
equation%
\begin{equation*}
s^{2m+2}+z_{1}s^{2m+1}+\cdots +z_{2m+2}=0\text{,}
\end{equation*}%
in which $\lambda _{1},\ldots ,\lambda _{m}$ are double roots. By the Vi\`{e}%
te's formulas,%
\begin{eqnarray*}
z_{1} &=&\left( -1\right) \left( 2\overset{m}{\underset{i=1}{\sum }}\lambda
_{i}+\lambda _{m+1}+\lambda _{m+2}\right) , \\
&&\vdots \\
z_{2m+2} &=&\left( -1\right) ^{2m+2}\lambda _{m+1}\lambda _{m+2}\overset{m}{%
\underset{i=1}{\prod }}\lambda _{i}^{2}.
\end{eqnarray*}
\end{proof}

$\left. {}\right. $

\begin{proof}[Proof of Necessity of Lemma 1]
We only need to show that there do not exist nonzero real numbers for (3)
when%
\begin{equation*}
B=\left[
\begin{array}{cccc}
\lambda & 1 &  &  \\
& \lambda & \ddots &  \\
&  & \ddots & 1 \\
&  &  & \lambda%
\end{array}%
\right] \in
\mathbb{R}
^{N\times N}
\end{equation*}%
where $N\geq 3$ and $\lambda \neq 0$. Reduction to absurdity: assume that
such nonzero real numbers exist for (3) with the above given $B$. From (3),%
\begin{equation}
\left( \frac{1}{u\left( L\right) }I+B\right) \left( \frac{1}{u\left(
L-1\right) }I+B\right) \cdots \left( \frac{1}{u\left( 1\right) }I+B\right)
\left( \frac{1}{u\left( 0\right) }I+B\right) =\overset{L}{\underset{k=0}{%
\prod }}\frac{1}{u\left( k\right) }I.
\end{equation}%
It follows%
\begin{eqnarray}
B^{L+1}+\overset{L}{\underset{k=0}{\sum }}\frac{1}{u\left( k\right) }%
B^{L}+\cdots +\overset{L}{\underset{k=0}{\sum }}\frac{u\left( k\right) }{%
\overset{2m+2}{\underset{j=0}{\prod }}u\left( j\right) }B &=&0\Rightarrow
\notag \\
B^{L}+\overset{L}{\underset{k=0}{\sum }}\frac{1}{u\left( k\right) }%
B^{L-1}+\cdots +\overset{L}{\underset{k=0}{\sum }}\frac{u\left( k\right) }{%
\overset{2m+2}{\underset{j=0}{\prod }}u\left( j\right) }I &=&0,
\end{eqnarray}%
which implies $L\geq N$. Otherwise the $\left( 1,L+1\right) $th entry in the
sum of the left side of matrix equation (34), which is determined by $B^{L}$%
, does not vanish. Next, let $w\left( k\right) =\frac{1}{u\left( k\right) }%
+\lambda $ for $k=0,1,\ldots ,L$. From (33),%
\begin{equation*}
\left( w\left( L\right) I+\bar{B}\right) \left( w\left( L-1\right) I+\bar{B}%
\right) \cdots \left( w\left( 1\right) I+\bar{B}\right) \left( w\left(
0\right) I+\bar{B}\right) =\overset{L}{\underset{k=0}{\prod }}\left( w\left(
k\right) -\lambda \right) I
\end{equation*}%
where%
\begin{equation}
\bar{B}=\left[
\begin{array}{cccc}
0 & 1 &  &  \\
& 0 & \ddots &  \\
&  & \ddots & 1 \\
&  &  & 0%
\end{array}%
\right] .  \notag
\end{equation}%
We have%
\begin{align}
& \bar{B}^{L+1}+\overset{L}{\underset{k=0}{\sum }}w\left( k\right) \bar{B}%
^{L}+\cdots +\overset{L}{\underset{k=0}{\sum }}\frac{\overset{L}{\underset{%
j=0}{\prod }}w\left( j\right) }{w\left( k\right) }\bar{B}+\overset{L}{%
\underset{k=0}{\prod }}w\left( k\right) I  \notag \\
& =\sum w\left( k_{0}\right) w\left( k_{1}\right) \cdots w\left(
k_{L-N+1}\right) \bar{B}^{N-1}+\sum w\left( k_{0}\right) w\left(
k_{1}\right) \cdots w\left( k_{L-N+2}\right) \bar{B}^{N-2}+  \notag \\
& \cdots +\overset{L}{\underset{k=0}{\sum }}\frac{\overset{L}{\underset{j=0}{%
\prod }}w\left( j\right) }{w\left( k\right) }\bar{B}+\overset{L}{\underset{%
k=0}{\prod }}w\left( k\right) I  \notag \\
& =\overset{L}{\underset{k=0}{\prod }}\left( w\left( k\right) -\lambda
\right) I
\end{align}%
since $\bar{B}^{j}=0$ for $j\geq N$. Therefore, from (35)%
\begin{eqnarray}
\sum w\left( k_{0}\right) w\left( k_{1}\right) \cdots w\left(
k_{L-N+1}\right) &=&0,  \notag \\
&&\vdots  \notag \\
\overset{L}{\underset{k=0}{\sum }}\frac{\overset{L}{\underset{j=0}{\prod }}%
w\left( j\right) }{w\left( k\right) } &=&0,  \notag \\
\overset{L}{\underset{k=0}{\prod }}w\left( k\right) &=&\overset{L}{\underset{%
k=0}{\prod }}\left( w\left( k\right) -\lambda \right) .
\end{eqnarray}%
Let%
\begin{eqnarray}
\overset{L}{\underset{k=0}{\sum }}w\left( k\right) &=&\left( -1\right) c_{1},
\notag \\
&&\underset{}{\vdots }  \notag \\
\sum w\left( k_{0}\right) w\left( k_{1}\right) \cdots w\left( k_{L-N}\right)
&=&\left( -1\right) ^{L-N+1}c_{L-N+1}.
\end{eqnarray}%
For the last equation in (36), by noting the other equations in (36) and the
equations in (37) we have%
\begin{eqnarray*}
\overset{L}{\underset{k=0}{\prod }}\left( w\left( k\right) -\lambda \right)
&=&\overset{L}{\underset{k=0}{\prod }}w\left( k\right) +\overset{L}{\underset%
{k=0}{\sum }}\frac{\overset{L}{\underset{j=0}{\prod }}w\left( j\right) }{%
w\left( k\right) }\left( -\lambda \right) +\cdots +\overset{L}{\underset{k=0}%
{\sum }}w\left( k\right) \left( -\lambda \right) ^{L}+\left( -\lambda
\right) ^{L+1} \\
&=&\overset{L}{\underset{k=0}{\prod }}w\left( k\right) +\left( -1\right)
^{L-N+1}c_{L-N+1}\left( -\lambda \right) ^{N}+\cdots +\left( -1\right)
c_{1}\left( -\lambda \right) ^{L}+\left( -\lambda \right) ^{L+1} \\
&=&\overset{L}{\underset{k=0}{\prod }}w\left( k\right) +\left( -1\right)
^{L+1}\left( c_{L-N+1}\lambda ^{N}+\cdots +c_{1}\lambda ^{L}+\lambda
^{L+1}\right) =\overset{L}{\underset{k=0}{\prod }}w\left( k\right) .
\end{eqnarray*}%
Hence,%
\begin{eqnarray}
c_{L-N+1}\lambda ^{N}+\cdots +c_{1}\lambda ^{L}+\lambda ^{L+1}
&=&0\Rightarrow  \notag \\
c_{L-N+1} &=&-\left( \lambda ^{L-N+1}+c_{1}\lambda ^{L-N}+\cdots
+c_{L-N}\lambda \right) .
\end{eqnarray}%
Now, let%
\begin{equation}
\overset{L}{\underset{k=0}{\prod }}w\left( k\right) =\left( -1\right)
^{L+1}c.
\end{equation}%
Combining (37), (36), (39) and noting (38), one can see from the Vi\`{e}te's
formulas\ that $w\left( 0\right) ,w\left( 1\right) ,$ $\ldots ,w\left(
L-1\right) ,w\left( L\right) $ are all the roots of the following$\ \left(
L+1\right) $th-degree equation%
\begin{equation}
s^{L+1}+c_{1}s^{L}+\cdots +c_{L-N}s^{N+1}-\left( \lambda
^{L-N+1}+c_{1}\lambda ^{L-N}+\cdots +c_{L-N}\lambda \right) s^{N}+c=0.
\end{equation}%
However, according to Lemma 5 in [19], eq. (40) must have complex roots.
This contradicts the fact that $u\left( 0\right) ,u\left( 1\right) ,\ldots
,u\left( L-1\right) ,u\left( L\right) $ are all real. Therefore, if there
exist nonzero real numbers to satisfy eq. (3), then $B$ does not have a
Jordan block with dimension greater than two in its Jordan canonical form.
\end{proof}

$\left. {}\right. $

\begin{proof}[Proof of Necessity of Theorem 2]
For necessity, we show that if $B$ is singular, or noncyclic, or has a
Jordan block with dimension greater than two in its Jordan canonical form,
then the system (2) is non-nearly controllable. Firstly, if $B$ is singular,
we can write%
\begin{equation*}
B=\left[
\begin{array}{cc}
\ast & \ast \\
\mathbf{0} & 0%
\end{array}%
\right]
\end{equation*}%
without loss of generality. Then, $u(k)$ loses the ability in controlling $%
x_{n}(k)$ and the system is non-nearly controllable.

Secondly, if $B$ is not cyclic, it has an eigenvalue that corresponds to at
least two Jordan blocks in its Jordan canonical form and can be simply
written as%
\begin{equation*}
\left[
\begin{array}{ccccc}
\ast & \ast &  &  &  \\
& \lambda &  &  &  \\
&  & \ast & \ast &  \\
&  &  & \lambda &  \\
&  &  &  & \ast%
\end{array}%
\right]
\end{equation*}%
without loss of generality. Then, for any initial state $\xi =\left[
\begin{array}{ccccc}
\cdots & \xi _{i} & \cdots & \xi _{j} & \cdots%
\end{array}%
\right] ^{T}$, it follows%
\begin{eqnarray*}
&&\overset{L}{\underset{k=0}{\prod }}\left( I+u\left( k\right) B\right) %
\left[
\begin{array}{ccccc}
\cdots & \xi _{i} & \cdots & \xi _{j} & \cdots%
\end{array}%
\right] ^{T} \\
&=&\left[
\begin{array}{ccccc}
\ast & \ast &  &  &  \\
& \overset{L}{\underset{k=0}{\prod }}\left( 1+u\left( k\right) \lambda
\right) &  &  &  \\
&  & \ast & \ast &  \\
&  &  & \overset{L}{\underset{k=0}{\prod }}\left( 1+u\left( k\right) \lambda
\right) &  \\
&  &  &  & \ast%
\end{array}%
\right] \left[
\begin{array}{c}
\vdots \\
\xi _{i} \\
\vdots \\
\xi _{j} \\
\vdots%
\end{array}%
\right] \\
&=&\left[
\begin{array}{ccccc}
\cdots & \overset{L}{\underset{k=0}{\prod }}\left( 1+u\left( k\right)
\lambda \right) \xi _{i} & \cdots & \overset{L}{\underset{k=0}{\prod }}%
\left( 1+u\left( k\right) \lambda \right) \xi _{j} & \cdots%
\end{array}%
\right] ^{T}\triangleq \left[
\begin{array}{ccccc}
\cdots & \eta _{i} & \cdots & \eta _{j} & \cdots%
\end{array}%
\right] ^{T},
\end{eqnarray*}%
from which we can see $\xi _{i}\eta _{j}=\xi _{j}\eta _{i}$. This implies
that some of the state variables in the terminal states are linearly
dependent. Thus, the system does not have a large controllable region and is
non-nearly controllable.

Finally, if $B$ has a Jordan block with dimension greater than two in its
Jordan canonical form, then the system is non-nearly controllable by Lemma 9.
\end{proof}

$\left. {}\right. $

\begin{proof}[Proof of Lemma 9]
From the analysis on system (20), we can see that system (25) is
controllable on%
\begin{equation*}
\left\{ \xi =\left[
\begin{array}{cccc}
\xi _{1} & 0 & \cdots & 0%
\end{array}%
\right] ^{T}\left\vert \text{ }\xi _{1}\neq 0\right. \right\} ,\text{ }%
\left\{ \xi =\left[
\begin{array}{ccccc}
\xi _{1} & \xi _{2} & 0 & \cdots & 0%
\end{array}%
\right] ^{T}\left\vert \text{ }\xi _{2}\neq 0\right. \right\}
\end{equation*}%
respectively. Now, if system (25) has a controllable region that contains
more than one state and contains such a state%
\begin{equation*}
\bar{\xi}=\left[
\begin{array}{cccccc}
\xi _{1} & \cdots & \xi _{j} & 0 & \cdots & 0%
\end{array}%
\right] ^{T}
\end{equation*}%
where $\xi _{j}\neq 0$ and $j\geq 3$, then, for another state $\eta $ ($\eta
\neq \bar{\xi}$) in this region, there exist nonzero real control inputs $%
u(0),u(1),\ldots ,u(L_{1}-1),u\left( L_{1}\right) $ such that $\bar{\xi}$ is
transferred to $\eta $ and nonzero real control inputs $v(0),v(1),\ldots
,v(L_{2}-1),v\left( L_{2}\right) $ such that $\eta $ is transferred to $\bar{%
\xi}$. As a result, we have%
\begin{equation*}
\prod\limits_{k=0}^{L_{1}+L_{2}+1}\left( I+u\left( k\right) B\right) \bar{\xi%
}=\prod\limits_{k=L_{1}+1}^{L_{1}+L_{2}+1}\left( I+u\left( k\right) B\right)
\eta =\prod\limits_{k=0}^{L_{2}}\left( I+v\left( k\right) B\right) \eta =%
\bar{\xi}
\end{equation*}%
where $u\left( L_{1}+1\right) =v(0),u\left( L_{1}+2\right) =v(1),\ldots
,u\left( L_{1}+L_{2}\right) =v(L_{2}-1),u\left( L_{1}+L_{2}+1\right)
=v\left( L_{2}\right) $, i.e.%
\begin{equation*}
\prod\limits_{k=0}^{L_{1}+L_{2}+1}\left( I+u\left( k\right) B\right) \left[
\begin{array}{cccccc}
\xi _{1} & \cdots & \xi _{j} & 0 & \cdots & 0%
\end{array}%
\right] ^{T}=\left[
\begin{array}{cccccc}
\xi _{1} & \cdots & \xi _{j} & 0 & \cdots & 0%
\end{array}%
\right] ^{T}.
\end{equation*}%
This implies%
\begin{equation}
\prod\limits_{k=0}^{L_{1}+L_{2}+1}\left( I+u\left( k\right) B_{1,\ldots
,j}\right) \left[
\begin{array}{ccc}
\xi _{1} & \cdots & \xi _{j}%
\end{array}%
\right] ^{T}=\left[
\begin{array}{cccc}
\Pi & \Sigma _{1} & \cdots & \Sigma _{j-1} \\
& \ddots & \ddots & \vdots \\
&  & \Pi & \Sigma _{1} \\
&  &  & \Pi%
\end{array}%
\right] \left[
\begin{array}{c}
\xi _{1} \\
\vdots \\
\xi _{j}%
\end{array}%
\right] =\left[
\begin{array}{c}
\xi _{1} \\
\vdots \\
\xi _{j}%
\end{array}%
\right]
\end{equation}%
where%
\begin{equation*}
\Pi =\prod\limits_{k=0}^{L_{1}+L_{2}+1}\left( 1+u\left( k\right) \lambda
\right) ,\text{ }\Sigma _{i}=\frac{1}{i!}\frac{d^{i}\prod%
\limits_{k=0}^{L_{1}+L_{2}+1}\left( 1+u\left( k\right) y\right) }{dy^{i}}%
\left\vert
\begin{array}{c}
\\
_{y=\lambda }%
\end{array}%
\right. \text{ for }i=1,\ldots ,j-1.
\end{equation*}%
From (41)%
\begin{eqnarray*}
\Pi \xi _{1}+\Sigma _{1}\xi _{2}+\cdots +\Sigma _{j-1}\xi _{j} &=&\xi _{1},
\\
&&\vdots \\
\Pi \xi _{j-1}+\Sigma _{1}\xi _{j} &=&\xi _{j-1}, \\
\Pi \xi _{j} &=&\xi _{j}.
\end{eqnarray*}%
Using the last equation in the above equations, we can deduce%
\begin{equation*}
\Pi =1,\text{ }\Sigma _{1}=0,\ldots ,\Sigma _{j-1}=0.
\end{equation*}%
Then, it results in%
\begin{equation}
\prod\limits_{k=0}^{L_{1}+L_{2}+1}\left( I+u\left( k\right) B_{1,\ldots
,j}\right) =I.
\end{equation}%
However, since $j\geq 3$, there do not exist such nonzero control inputs to
satisfy (42) according to Lemma 1. Therefore, system (25) can only have the
states of which the latter $\left( n-2\right) $ entries are all zero in its
controllable regions that contain more than one state.
\end{proof}

\end{document}